\documentclass[runningheads]{llncs}
\usepackage[T1]{fontenc}

\usepackage{amsmath,amssymb,amsfonts,stmaryrd,extarrows,dsfont}
\usepackage{graphicx}
\usepackage{complexity}
\usepackage{graphicx}
\usepackage{todonotes}
\usepackage{bbold}
\usepackage{tikz}
\usepackage{subfigure}
\usepackage{url}
\usepackage{enumerate}
\usepackage{robsafety}
\usepackage{wrapfig}
\usepackage{multirow}
\usepackage{appendix}
\usepackage{thm-restate}
\usepackage{tikz}
\usetikzlibrary{decorations,arrows,shapes,automata, patterns,calc,positioning}
\usepackage{color}

\newcommand{\pspace}[0]{\mathsf{PSPACE}}
\renewcommand\paragraph[1]{\smallskip\noindent\textbf{#1.}}

\newcommand\ignore[1]{}

\begin{document}
\title{
Controller Synthesis in Timed B\"uchi Automata: Robustness and Punctual Guards
}
\titlerunning{Robustness and Punctual Guards}
\author{Benoît Barbot\inst{1} \and
Damien Busatto-Gaston\inst{1}\and\\
Catalin Dima\inst{1}\and
Youssouf Oualhadj\inst{1}}
\authorrunning{Barbot et. al.}
\institute{Univ Paris Est Creteil, LACL, F-94010 Creteil, France\\
\email{benoit.barbot@u-pec.fr}\\
\email{damien.busatto-gaston@u-pec.fr}\\
\email{catalin.dima@u-pec.fr}\\
\email{youssouf.oualhadj@u-pec.fr}}
\maketitle              %
\begin{abstract}
We consider the synthesis problem on timed automata with B\"uchi objectives, where delay choices made by a controller are subjected to small perturbations.
Usually, the controller needs to avoid punctual guards, such as testing the equality of a clock to a constant.
In this work, we generalize to a robustness setting that allows for punctual transitions in the automaton to be taken by controller with no perturbation.
In order to characterize cycles that resist perturbations in our setting,
we introduce a new
structural requirement on the reachability relation along an accepting cycle of the automaton.
This property is formulated on the %
region abstraction, and generalizes the existing characterization of winning cycles in the absence of punctual guards.
We show that the problem remains within $\pspace$ despite the presence of punctual guards.

\end{abstract}

\section{Introduction}

The design and verification of reactive systems usually requires the support of formal techniques to ensure the correction of the underlying model. 
The controller synthesis approach ensures a correct-by-design system since it supplies the designer of the system with formal tools ensuring that the system behaves according to a given specification.
These formal tools consist in a game theoretic modeling where the reactive system is decomposed into controllable and uncontrollable behaviors. Ensuring the correctness of the system boils down to designing a sequence of controllable moves such that the specification holds under these moves regardless of the uncontrollable move that might be taken by the system.
Thanks to the game theoretic toolbox, the sequence of controllable moves is automatically computed as a strategy in this game that ensures the correction of the designed systems. 

When designing a reactive system with timing constraints, the formalism of choice is the one of \emph{timed automata}.
In a nutshell, the reactive system is modelled as a timed automaton over which two entities will compete. 
The first entity (the controller) aims at enforcing an already agreed upon specification while the second one aims at preventing the controller from achieving it. This is known as timed games~\cite{MalerPS95,AsarinMPS98,CassezDFLL05}.  

That being said, timed automata are nothing but a mere mathematical tool that might exhibit unrealistic behaviors. 
They may allow models that rely on infinitely precise behaviors to satisfy a specification, which can sometimes be an unrealistic assumption. 
A series of works focused on integrating some degree of robustness in the semantics of timed automata, i.e.~enforcing behaviors resisting small 
perturbations in the evolution of the clock variables ~\cite{AB-formats11,Puri00,SBMR-concur13,GuptaHenzingerJagadeesan97,WulfDMR04}.

In particular, the controller synthesis problem for timed automata under \emph{conservative semantics of perturbations} and with a Büchi objective is known to be $\PSPACE$-complete \cite{SBMR-concur13,Sankur13} (see also \cite{BMRS19} for a symbolic approach and \cite{OualhadjRS14} for a probabilistic extension).
In this setting, a game is played between $\cont$ and $\pert$ on a timed automaton, with $\cont$ choosing delays and transitions to follow, and $\pert$ modifying every delay by a small amount. %
In particular, $\cont$ must make sure that any transition he picks is still available after the deviated time elapse.
The controller synthesis problem asks for the existence of a $\cont$ strategy in this game that guarantees the objective against any decisions made by $\pert$. 

The perturbation semantics naturally abstract imprecision in the measuring of time and the difficulty of enforcing a precise passage of time between consecutive events. 
To this end, the work in \cite{SBMR-concur13} %
rules out the possibility for $\cont$ to choose \emph{punctual} transitions,
i.e. transitions which can be taken after a unique delay, since 
any perturbation to this delay would disable the transition.
Such transitions correspond to guards in which at least one clock must be tested to have some exact value, for example $x = 2$.

However exact constraints might still occur in models of timed systems \cite{Rodriguez-NavasP13},
e.g.~when some controllers might have sufficiently reliable internal clocks. 
In such systems, some degree of unreliability in the measurement of time is still needed to model other
 uncontrollable components incorporating clocks that might be subjected to perturbations. 
A natural problem would then be to synthesize controllers in a system described as the synchronous composition of reliable components where delays and clock measurements are considered exact and unreliable components where perturbations are applied. 
This does not fit the setting of \cite{SBMR-concur13},
as some transitions would not be subjected to perturbations. At the heart of the issue lie punctual transitions derived from reliable components of the system, that could make their way to the resulting timed automaton despite being forbidden under classical perturbation semantics. 
In this paper, we forgo the compositional framework and directly allow $\cont$ to choose punctual transitions in the robust controller synthesis problem, in which case his choice of delay is not subjected to perturbation.
Our setting, in which edges of the automaton are either reliable and punctual, or unreliable and non-punctual, 
paves the way
towards the more general setting where even some non-punctual transitions could be considered reliable.

We show that this mixing of exact transitions and transitions under perturbation can be analyzed using the 
the so-called (folded) orbit graphs formalism from \cite{AB-formats11,Puri00,SBMR-concur13,DBLP:conf/formats/Stainer12}. 
Our algorithm that checks the existence of a robust controller for a Büchi objective 
involves finding, in the region automaton, a reachable cyclic path satisfying the Büchi condition whose folded orbit graph is a \emph{cluster graph}, i.e., a disjoint union of complete graphs.
The proof that such reachable cyclic paths %
can be robustly repeated 
involves, similarly with \cite{SBMR-concur13,OualhadjRS14},
showing that, starting from any clock valuation $\nu$, 
and no matter what valuation $\nu'$ is reached %
as a cycle under perturbation is followed,
$\cont$ can enforce a run from $\nu'$ to a neighborhood of $\nu$. Therefore, $\cont$ ensures the robustness of his behaviour as he can compensate for any deviation.

In particular, this involves asking the reachability relation from valuation to valuation along a cyclic path of the automaton to be complete on sets of valuations.
But, unlike in the case of \cite{SBMR-concur13}, where these sets form the standard region partition, in our case they will form a finer partition of each region into subsets that we call \emph{slices}.

In this paper, we start by defining in Section~\ref{sec:partRob} our perturbation semantics in the presence of punctual guards, then in Section~\ref{sec:Prel} we recall the folded orbit graphs that represent the reachability relation along a cycle.
In Section~\ref{sec:robust-path}, we show that robustly reaching a state by following some path of the automaton can be guaranteed under a simple syntactic assumption on the guards that are traversed.
In Section~\ref{sec:partRobCharc}, we define a class of folded orbit graphs (so-called cluster graphs), and state our main result: the cycles that can be iterated under our perturbed semantics exactly correspond to those in this class.
In Section~\ref{sec:Slices}, we introduce the slice partition induced by a cluster folded orbit graph in order to represent sets of valuations with a complete reachability relation. 
Finally, in Section~\ref{sec:proofMainThm} we prove our main result by adapting the reasoning of \cite{SBMR-concur13} to slices in the presence of punctual transitions.

\section{Partial robustness semantics}
\label{sec:partRob}
\paragraph{Timed automata}
Given a finite set of clocks $\Clocks$, we~call \emph{valuations} the
elements of~$\Realnn^\Clocks$. For a subset $R\subseteq \Clocks$
and a valuation~$\nu$, ${\nu[R:=0]}$~is the valuation defined by 
${\nu[R:=0](x) = \nu(x)}$ for $x \in \Clocks \setminus R$ and 
${\nu[R:=0](x) = 0}$ for $x \in R$. Given $d \in \Realnn$ 
and a valuation $\nu$,
the valuation $\nu+d$ is called a time-successor of $\nu$ and is defined by $(\nu+d)(x) = \nu(x)+d$ for all $x\in
\Clocks$. We~extend these operations to sets of valuations in the
obvious way, and write $\mathbf{0}$ for the valuation that assigns~$0$ to
every clock.
We will consider the usual $d_\infty$ metric on $\Real^\Clocks$, defined as
$d_\infty(\nu,\nu') = \max_{x \in \Clocks}|\nu(x) - \nu'(x)|$,
and the Manhattan norm $\|\vec{v}\|_1$ of a vector $\vec{v}=(v_1,\dots,v_k)$, defined as $\sum_{i=1}^k v_i$.
We say that a valuation $\nu$ is bounded by $\ClockBound>0$ if %
$\nu\in \Valuations$.

An atomic clock constraint is a formula of the form $k \preceq x \preceq' l$
or $k \preceq x - y \preceq' l$ where $x,y \in \Clocks$, $k,l \in
\mathbb{Z}\cup\{-\infty,\infty\}$ and ${\mathord\preceq,\mathord\preceq' \in
\{\mathord<,\mathord\leq\}}$.  A~\emph{guard} $g$ is a conjunction of
atomic clock constraints.
A~valuation~$\nu$ satisfies~$g$, denoted $\nu \models g$,
if all constraints are satisfied when each $x\in \Clocks$ is replaced
with~$\nu(x)$.
We~write $\Guards$ for the set of guards built on~$\Clocks$.
A guard $g$ is called \emph{non-punctual} if there are at least two valuations $\nu$ and $\nu+d$ with $d\in\Realnn$ that satisfy $g$.
Conversely, a guard $g$ is \emph{punctual} if $\nu\models g\land d>0$ implies  $\nu+d\not\models g$.

A (bounded)~\emph{timed automaton} $\TA$ is a tuple $(\Locs,\Clocks,\ClockBound,\ell_0,E)$, where $\Locs$
is a finite set of locations, $\Clocks$~is a finite set of clocks, 
$\ClockBound\in\Natp$ is an upper bound for clocks,
$E\subseteq \mathcal{L} \times \Guards \times 2^\Clocks
\times \mathcal{L}$ is a set of edges, and $\ell_0\in \Locs$ is the
initial location.  An~edge $e = (\ell,g,R,\ell')$ is also written as
$\ell \xrightarrow{g, R} \ell'$.  
\begin{figure}[t!]
  \begin{center}
  \vspace{-.7cm}
  \begin{minipage}{.55\textwidth}
    \centering
    \begin{tikzpicture}[node distance=2.5cm,auto, scale = .7]
      \tikzstyle{every state}=[thick, circle,minimum size=17pt,inner sep=0pt]
      \node (preA) at (-1,0) {};
      \node[state] at (0,0) (A) {$\ell_0$};
      \node[state, right of=A,color=red!60!lightgray] (B) {$\ell_1$};
      \node[accepting,state,right of=B,color=blue!50] (C) {$\ell_2$};
      \everymath{\scriptstyle}
      \path[-latex'] 
      (preA) edge (A)
      (A) edge node[above] {$x<1, y:=0$} (B)
      (C) edge[bend left] node[below] {$\scriptstyle y = 2, y := 0$} (B)
      (B) edge[bend left] node[above] {$\scriptstyle x <  2, x:=0$} (C);
    \end{tikzpicture}
    \end{minipage}
    \hfill
    \begin{minipage}{.35\textwidth}
    \centering
  \begin{tikzpicture}[scale = .7]
    \foreach \x in {1,2,3} {\everymath{\scriptstyle}
    \draw (\x,0) -- +(-90:1mm) node[below] {$\x$};
    \draw[dotted, thin] (\x,0) -- +(0,3.2);
    \draw (0,\x) -- +(180:1mm) node[left] {$\x$};
    \draw[dotted, thin] (0,\x) -- +(3.2,0);}
    \everymath{\scriptstyle}
    \draw[latex'-latex', thin] (3.2,0) 
    node[right] {$x$} -- (0,0) 
    node[below] {$0$} node [left] {$0$} -- (0,3.2) 
    node[above] {$y$};
    \path[use as bounding box] (0,-1);
    \fill[] (0,0) circle(1.7pt);
    \draw[fill, opacity=.3, line width=1.7pt, shorten <=0.1cm, shorten >=0.1cm] (0,0) -- (1,1);
    \draw[fill, opacity=1, color=red!50, line width=2.2pt, shorten <=0.1cm, shorten >=0.1cm] (0,0) -- (1,0);
    \fill[opacity=.3,rounded corners=1pt,color=red!50] (0.19,0.1) -- (0.9,0.81) -- (0.9,0.1) -- cycle;
    \fill[opacity=.3,rounded corners=1pt,color=red!50] (1.1,0.9) -- (1.81,0.9) -- (1.1,0.19) -- cycle;
    \fill[opacity=1,rounded corners=1pt,color=red!50] (1.19,1.1) -- (1.9,1.81) -- (1.9,1.1) -- cycle;
    
    \draw[fill, opacity=.3, line width=2.2pt, shorten <=0.1cm, shorten >=0.1cm,color=red!50] (1,0) -- (2,0);
    \fill[opacity=.3,rounded corners=1pt,color=red!50] (1.19,0.1) -- (1.9,0.81) -- (1.9,0.1) -- cycle;

    \draw[fill, opacity=.3, line width=2.2pt, shorten <=0.1cm, shorten >=0.1cm,color=blue!50] (0,0) -- (0,1);
    \draw[fill, opacity=1, line width=2.2pt, shorten <=0.1cm, shorten >=0.1cm,color=blue!50] (0,1) -- (0,2);

    \fill[opacity=.3,rounded corners=1pt,color=blue!50] (0.1,0.9) -- (0.81,0.9) -- (0.1,0.19) -- cycle;
    \fill[opacity=.3,rounded corners=1pt,color=blue!50] (0.1,1.9) -- (0.81,1.9) -- (0.1,1.19) -- cycle;
    \fill[opacity=.3,rounded corners=1pt,color=blue!50] (1.1,1.9) -- (1.81,1.9) -- (1.1,1.19) -- cycle;
    \fill[opacity=.3,rounded corners=1pt,color=blue!50] (0.19,1.1) -- (0.9,1.81) -- (0.9,1.1) -- cycle;

    \draw[fill, opacity=1, line width=2.2pt, shorten <=0.1cm, shorten >=0.1cm,color=blue!50] (0,2) -- (1,2);
    \draw[fill, opacity=.3, line width=2.2pt, shorten <=0.1cm, shorten >=0.1cm,color=blue!50] (1,2) -- (2,2);

  \end{tikzpicture}   
  \end{minipage}
  \vspace{-.4cm}
    \caption{\label{fig:taEx}A timed automaton, and some of its reachable regions.}
  \end{center}
  \vspace{-.6cm}
\end{figure}
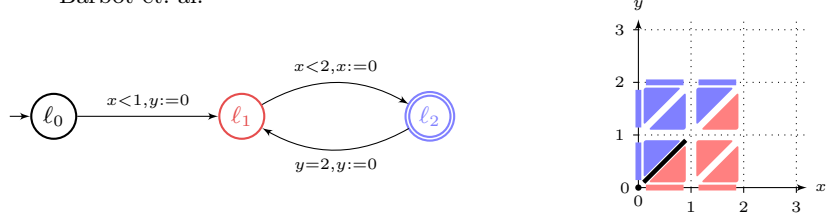

A configuration is a
pair $q=(\ell,\nu)\in \Locs\times \Realnn^\Clocks$.
The set of possible behaviors of a timed automaton can be described by the set
of its runs, as follows.
An infinite~\emph{run} of~$\TA$ is a sequence $q_1e_1q_2e_2\ldots$ where $q_i \in
\Locs\times \Valuations$,
and writing $q_i = (\ell,\nu)$, either $e_i \in \Realnn$, in which case $q_{i+1}
= (\ell,\nu + e_i)$ and $e_i$ is called a delay transition, or
$e_i =(\ell,g,R,\ell')\in E$, in which case 
$\nu \models g$, $q_{i+1} = (\ell',\nu[R:=0])$ and $e_i$ is called an edge transition.
A delay transition is sometimes denoted $(\ell,\nu)\xrightarrow{d}(\ell,\nu+d)$,
and an edge transition is sometimes denoted $(\ell,\nu)\xrightarrow{g,R}(\ell,\nu[R:=0])$.
A finite run of length $k$ is a sequence $q_1e_1\ldots q_ke_kq_{k+1}$ defined similarly. %
Whenever $\rho$ is a finite run that ends in $(\ell,\nu)$ and $\rho'$ is a finite or infinite run starting from $(\ell,\nu)$, we let $\rho\rho'$ denote the concatenations of the two.

\paragraph{Game semantics}
In order to define the robust controller synthesis problem, we introduce a game played on a timed automaton $\TA$.
Two players called \emph{$\cont$} and \emph{$\pert$} compete in an arena induced by 
$\TA$ and a parameter $\delta > 0$ in a two-player zero-sum game $\thegame$.
The interaction between $\cont$ and $\pert$ in $\thegame$ follows the following rules to build an infinite run:
given an initial configuration $(\ell,\nu)$, or a finite run ending in $(\ell,\nu)$,
\begin{itemize}
\item either $\cont$ picks a delay $d\geq 0$ and an edge $e=(\ell,g,R,\ell')$ with a \emph{punctual} guard $g$ to extend the run so that $\nu+d\models g$, in which case the run is extended with delay $d$ and edge $e$ without any perturbation,
\item or $\cont$ picks a delay $d\geq \delta$ and an edge $e=(\ell,g,R,\ell')$ with a \emph{non-punctual} guard $g$ to extend the run, so that $g$ is satisfied after any delay in the set $[d-\delta,d+\delta]$. In this case, $\pert$ chooses an actual delay $d' \in [d-\delta,d+\delta]$ to extend the run, after which the edge $e$ is taken. 
\end{itemize}
It is then again $\cont$'s turn to play from the new configuration.
We say that a run is \emph{well-formed} if it is compatible with these semantics, so that it starts with a delay transition, alternates between delay transitions and edge transitions, and ends with an edge transition.
Note that the concatenation of well-formed runs is well-formed as well, and that a well-formed run can be decomposed as the concatenation of atomic well-formed runs $(\ell,\nu)\xrightarrow{d}(\ell,\nu')\xrightarrow{g,R}(\ell',\nu'')$.

Formally, the state space of $\thegame$ is partitioned into the
states where it is $\cont$'s turn to make a move
$\Locs \times \Valuations$, and the %
states where it is $\pert$'s turn to make a move $\Locs \times \Valuations\times \Realnn \times E$. 
A play over the arena induced by $\thegame$ is an infinite alternation between states of $\cont$ and states of $\pert$ 
that forms a run of the automaton.
A strategy for $\cont$ in  $\thegame$ is a function 
$\sigmacont$ that assigns to every finite play ending in a state of $\cont$ a pair $(d, e)$ where $d$ is a delay and $e$ an edge of $\TA$.
A strategy for $\pert$ in $\thegame$ is function $\sigmapert$ that assigns to every finite play ending in a state of $\pert$ a perturbation
in $[-\delta, \delta]$.
Once a pair of strategies $(\sigmacont, \sigmapert)$ is fixed, we denote by $\outcome(\sigmacont, \sigmapert)$ the unique run induced over $\TA$ that starts from $(\ell_0,\vec{0})$ and follows them.
Notice that this run is not necessarily infinite, since some $\cont$ states may have no legal moves. %

A Büchi objective is given by a subset of location $\mathsf{Buchi} \subseteq \Locs$. 
We say that an infinite run $\rho$ in $\TA$ satisfies a Büchi objective $\mathsf{Buchi}$ if the set of locations visited infinitely often along $\rho$ contains states from $\mathsf{Buchi}$.
We say that $\cont$ wins $\thegame$ for the objective $\mathsf{Buchi}$ if $\cont$ has a winning strategy $\sigmacont$ 
so that for any strategy $\sigmapert$ played by $\pert$, 
$\outcome(\sigmacont, \sigmapert)$
is an infinite run that satisfies $\mathsf{Buchi}$.

Finally, we define the \emph{robust controller synthesis problem} as follows: given a timed automaton $\TA$ equipped with a Büchi objective $\mathsf{Buchi}$, does there exists an amplitude $\delta > 0$ small enough so that $\cont$ wins the resulting game $\thegame$?

\section{Preliminaries}
\label{sec:Prel}
\paragraph{The region abstraction}
With respect to the set $\Clocks$ of clocks and an upper bound $\ClockBound\in\Natp$ on %
clocks,
we partition the set $\Valuations$ of valuations into finitely many regions \cite{AD-tcs94}. We denote by $\Regions$ this set of (bounded) regions. Each region is characterised by a pair $(\iota, \beta)$ where $\iota\colon \Clocks\to [0,\ClockBound)\cap \Nat$
and $\beta$ is an ordered partition of~$\Clocks$ into subsets $\beta_0\uplus \beta_1\uplus \cdots \uplus \beta_m$ (with $m\geq 0$), where $\beta_0$ can be empty but $\beta_i\neq\emptyset$ for $1\leq i\leq m$. A valuation $\nu$ of $\Valuations$ belongs to the region characterised by $(\iota, \beta)$ if:
\begin{itemize}
\item for all $x\in\Clocks$,
$\iota(x) =\lfloor \nu(x) \rfloor$, where $\lfloor \nu(x) \rfloor$ is the integral part of $\nu(x)$;
\item for all $x\in \beta_0$, $\mathsf{fract}(\nu(x))=0$, where $\mathsf{fract}(\nu(x))=\nu(x)-\lfloor \nu(x) \rfloor$ is the fractional part of $\nu(x)$;
\item for all $1\leq i< j\leq m$, for all $x,y\in \beta_i$
and all $z\in \beta_j$, 
$0<\mathsf{fract}(\nu(x))$, $\mathsf{fract}(\nu(x)) = \mathsf{fract}(\nu(y))$ and $\mathsf{fract}(\nu(x)) < \mathsf{fract}(\nu(z))$.
\end{itemize}
\noindent A region $r$ characterized by $(\iota, \beta_0\uplus \cdots \uplus \beta_m)$ is said to be of dimension $m$.

\begin{remark}
Intuitively, a region of dimension $m$ is a polytope of dimension $m$ containing every valuation $\nu$ such that the integer part of its coordinates is the integer valuation $\iota$, and such that the fractional part of its coordinate is ordered according to the sequence $\beta_0<\beta_1<\cdots<\beta_m$, with clocks in the same $\beta_i$ having the same fractional part and those in $\beta_0$ having fractional part $0$.
\end{remark}

A region $r$ can be described by a system of equations that forms a guard $g$ of $\Guards$. Conversely, a guard intersected with $\Valuations$ can be seen as a finite union of regions in $\Regions$.
We say that a region $r$ is non-punctual if it contains at least two distinct time-successor valuations. It is punctual otherwise.

The set $\mathsf{Reg}(\Clocks,\ClockBound)$ partitions $\Valuations$ into 
regions
that represent equivalence classes of the time-abstract bisimulation~\cite{AD-tcs94}.

A region $r'$ is a time-successor of a region $r$ if for each valuation $\nu\in r$ there is a delay $d\geq 0$ so that $\nu+d\in r'$.
Similarly, given a region $r$ and a set of clocks $R$, the region $r[R:=0]$ is the region such that for all valuations $\nu\in r$, $\nu[R:=0]\in r[R:=0]$.
Finally, we write $r\models g$ if every valuation in $r$ satisfies $g$.

A pair $(\ell,r)$ with $\ell\in L$ and $r$ a region is called a region state.
A \emph{region path} $\pi$ is a finite or infinite sequence of delay transitions
$(\ell,r)\xrightarrow{\mathsf{delay}}(\ell,r')$ such that $r'$ is a time-successor of $r$
and edge transitions $(\ell,r)\xrightarrow{g,R}(\ell',r[R:=0])$ such that there is an edge $\ell\xrightarrow{g,R}\ell'$ in $E$ with $r\models g$.
We define the length of a finite region path and the concatenation of region paths as expected.
Similarly, a region path is well-formed if it is compatible with the game semantics, so that it starts with a delay transition, alternates between delays and edges, and ends with an edge transition.
A run $\rho$ is said to follow a region path $\pi$ if it contains the same sequence of delay and edge transitions, and every valuation visited along $\rho$ belongs to the corresponding region at the same step in $\pi$.

A \emph{region cycle} around a region state $(\ell,r)$ is a finite region path $\pi$ that starts and ends in the same region state $(\ell,r)$. We often omit $\ell$ and say that $\pi$ is a cycle around $r$.
A \emph{region lasso} $\pi_0\pi^\omega$ around a region state $(\ell,r)$ is an infinite region path described as a finite region path $\pi_0$ that ends in $(\ell,r)$, followed by a region cycle $\pi$ around $(\ell,r)$ that is iterated forever.
In this definition, we allow for $\pi_0$ to be empty (so that $\pi^\omega$ is a valid lasso), but $\pi$ has non-zero length. %

In the absence of perturbations, finding a winning run for a Büchi objective $\mathsf{Buchi}$ amounts to finding a well-formed region lasso $\pi_0\pi^\omega$ around some region state $(\ell,r)$ where $\ell$ is in $\mathsf{Buchi}$~\cite{AD-tcs94}.
We call such lassos \emph{winning lassos}.

\paragraph{Corners}
The \emph{corners} of a region $r$ of dimension $m$ characterized by $(\iota, \beta)$ with $\beta=\beta_0\uplus \cdots \uplus \beta_m$ are the following $m+1$ valuations: for each $0\leq i\leq m$, let $c_i$ be the corner such that $\forall 0\leq j\leq m-i, \forall x\in \beta_j, c_i(x)=\iota(x)$, and $\forall m-i< j\leq m$, $\forall x\in \beta_j, c_i(x)=\iota(x)+1$.
Let $\mathcal C$ denote the set of corners $\{c_0,\dots c_m\}$.

\begin{remark}
Intuitively, $\mathcal C$ contains the valuations of integer coordinates that form the corners of the polyhedron $r$.
In particular, they are totally ordered in Manhattan norm: $\|c_0\|_1<\|c_1\|_1<\dots<\|c_m\|_1$.
\end{remark}

Let $\nu$ be a valuation in $r$ of dimension $m$. From~\cite{Puri00}, we know that $\nu$ can always  be expressed as a unique convex combination of the corners $\{c_0,\dots c_m\}$ of $r$, so that there is a unique weight vector
$\vecl=(\lambda_0,\dots,\lambda_m)\in(0,1]^{m+1}$
with $\|\vecl\|_1=1$ such that $\nu=\sum_{i=0}^m \lambda_i c_i$.
Conversely, every combination of corners $\sum_{i=0}^m \lambda_i c_i$ where $\vecl=(\lambda_0,\dots,\lambda_m)\in(0,1]^{m+1}$ and $\|\vecl\|_1=1$ is a valuation in $r$.
\begin{definition}%
    A valuation $\nu$ in a region $r$ of dimension $m$ is said to have
    \emph{corner weights} $\vecl$ if $\vecl=(\lambda_0,\dots,\lambda_m)\in(0,1]^{m+1}$,
    $\|\vecl\|_1=1$ and $\nu=\sum_{i=0}^m \lambda_i c_i$.
\end{definition}

\paragraph{Orbit graphs}
Let $(\mathsf{step}_i)_{i\in\Nat}$ be a sequence of distinct fresh symbols.
An orbit graph $\gamma$ is a finite acyclic graph defined on a sequence of regions $r_0,\dots r_k$, so that the vertices of the graph are the pairs $(\mathsf{step}_i,c)$ where $0\leq i\leq k$ and $c$ ranges over the corners of $r_i$, and so that every edge in $\gamma$ is of the shape $(\mathsf{step}_i,c)\rightarrow (\mathsf{step}_{i+1},c')$.
Given an orbit graph $\gamma$ defined on a sequence of regions $r_0,\dots r_k$, and an orbit graph $\gamma'$ defined on a sequence of regions $r_0',r_1',\dots r_{l}'$ with $r_k=r'_0$, the concatenation of $\gamma$ and $\gamma'$ is the orbit graph $\gamma\gamma'$ defined on the sequence of regions $r_0,\dots r_k, r_1',\dots r_{l}'$
so that every edge $(\mathsf{step}_i,c)\rightarrow (\mathsf{step}_{i+1},c')$ of $\gamma$ is in $\gamma\gamma'$ and so that every edge $(\mathsf{step}_i,c)\rightarrow (\mathsf{step}_{i+1},c')$ of $\gamma'$ is in $\gamma\gamma'$ as $(\mathsf{step}_{i+k},c)\rightarrow (\mathsf{step}_{i+k+1},c')$.
In the context of an orbit graph $\gamma$ defined on a sequence of regions $r_0,\dots r_k$, we let $\mathsf{start}$ and $\mathsf{end}$ denote the symbols $\mathsf{step}_0$ and $\mathsf{step}_k$, respectively, so that a path going through all of $\gamma$ starts in a vertex $(\mathsf{start},c)$ and ends in a vertex $(\mathsf{end},c')$.

Let $t=(\ell,r)\xrightarrow{\mathsf{delay}}(\ell,r')$ be a delay transition from a region $r$ of corners $\mathcal C$ to a time-successor region $r'$ of corners $\mathcal C'$.
The orbit graph of $t$ is the orbit graph defined on $r,r'$
that contains every edge $(\mathsf{start},c)\rightarrow(\mathsf{end},c')$ such that $c'$ is a time-successor of $c$.
Let $t=(\ell,r)\xrightarrow{g,R}(\ell',r')$ be an edge transition from a region $r$ of corners $\mathcal C$ to a region $r'$ of corners $\mathcal C'$ with $r\models g$ and $r'=r[R:=0]$.
The orbit graph of $t$ is the orbit graph defined on $r,r'$ 
that contains every edge $(\mathsf{start},c)\rightarrow(\mathsf{end},c')$ such that $c'=c[R:=0]$.

Let $\pi$ be a finite region path. The \emph{orbit graph} of $\pi$, denoted $\gamma(\pi)$, is defined as the concatenation of the orbit graphs of every delay transition and edge transition in $\pi$.
In \figurename~\ref{fig:orbit-graph}, we depict a region path and its associated orbit graph.

\begin{remark}
    Intuitively, paths from corners to corners in $\gamma(\pi)$ represent the reachability relation along $\pi$ of valuations arbitrarily close to these corners. As every valuation in a region $r$ can be seen as a convex combination of the corners of $r$, $\gamma(\pi)$ represents as a finite graph the entire reachability relation along $\pi$: 
    If $\gamma(\pi)$ is interpreted as an interval Markov chain (the set of all Markov chains with this graph as support), then
    the runs that follow $\pi$ can be described as a Markov chain refining $\gamma(\pi)$, such that
    the corner weights of the starting valuation in the run describe the initial-state distribution,
    and the transition probabilities encode the flow of these corner weights along the run.
    This is formalised in Appendix~\ref{app:fog}.
\end{remark}

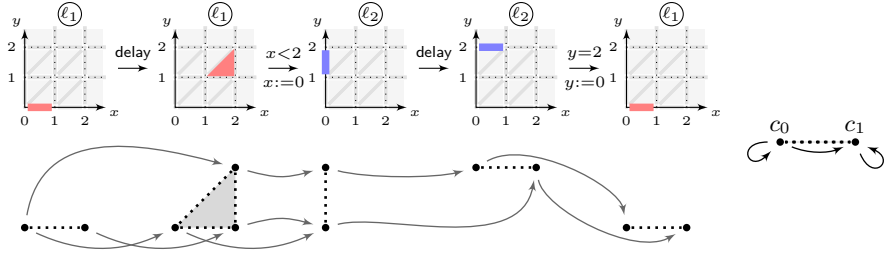
\begin{figure}[!t]
\def\yellow{black!10!white}
\def\red{black!30!white}
\def\green{black!10!white}
\def\blue{black!30!white}
\def\drawregions#1{%
  \fill[white!60!\yellow] (1.2,1.1) -- (1.9,1.1) -- (1.9,1.8) -- cycle;
  \fill[white!60!\yellow] (.2,1.1) -- (.9,1.1) -- (.9,1.8) -- cycle;
  \fill[white!60!\yellow] (1.2,.1) -- (1.9,.1) -- (1.9,.8) -- cycle;
  \fill[white!60!\yellow] (.2,.1) -- (.9,.1) -- (.9,.8) -- cycle;
  \fill[white!60!\yellow] (1.1,1.2) -- (1.1,1.9) -- (1.8,1.9) -- cycle;
  \fill[white!60!\yellow] (1.1,.2) -- (1.1,.9) -- (1.8,.9) -- cycle;
  \fill[white!60!\yellow] (.1,1.2) -- (.1,1.9) -- (.8,1.9) -- cycle;
  \fill[white!60!\yellow] (.1,.2) -- (.1,.9) -- (.8,.9) -- cycle;
  \draw[white!60!\red,line width=.5mm] (0.1,0) -- (0.9,0);
  \draw[white!60!\red,line width=.5mm] (1.1,0) -- (1.9,0);
  \draw[white!60!\red,line width=.5mm] (0.1,1) -- (0.9,1);
  \draw[white!60!\red,line width=.5mm] (1.1,1) -- (1.9,1);
  \draw[white!60!\red,line width=.5mm] (0.1,2) -- (0.9,2);
  \draw[white!60!\red,line width=.5mm] (1.1,2) -- (1.9,2);
  \draw[white!60!\red,line width=.5mm] (2,1.1) -- (2,1.9);
  \draw[white!60!\red,line width=.5mm] (2,0.1) -- (2,0.9);
  \draw[white!60!\red,line width=.5mm] (1,1.1) -- (1,1.9);
  \draw[white!60!\red,line width=.5mm] (1,0.1) -- (1,0.9);
  \draw[white!60!\red,line width=.5mm] (0,1.1) -- (0,1.9);
  \draw[white!60!\red,line width=.5mm] (0,0.1) -- (0,0.9);
  \draw[white!60!\red,line width=.5mm] (0,2.1) -- (0,2.6);
  \draw[white!60!\red,line width=.5mm] (1,2.1) -- (1,2.7);
  \draw[white!60!\red,line width=.5mm] (2,2.1) -- (2,2.7);
  \draw[white!60!\red,line width=.5mm] (2.1,0) -- (2.6,0);
  \draw[white!60!\red,line width=.5mm] (2.1,1) -- (2.7,1);
  \draw[white!60!\red,line width=.5mm] (2.1,2) -- (2.7,2);
  \draw[white!60!\red,line width=.5mm] (.1,.1) -- (.9,.9);
  \draw[white!60!\red,line width=.5mm] (1.1,.1) -- (1.9,.9);
  \draw[white!60!\red,line width=.5mm] (.1,1.1) -- (.9,1.9);
  \draw[white!60!\red,line width=.5mm] (1.1,1.1) -- (1.9,1.9);
  \fill[white!60!\blue] (0,0) circle(.6mm);
  \fill[white!60!\blue] (1,0) circle(.6mm);
  \fill[white!60!\blue] (0,1) circle(.6mm);
  \fill[white!60!\blue] (1,1) circle(.6mm);
  \fill[white!60!\blue] (2,0) circle(.6mm);
  \fill[white!60!\blue] (0,2) circle(.6mm);
  \fill[white!60!\blue] (2,1) circle(.6mm);
  \fill[white!60!\blue] (1,2) circle(.6mm);
  \fill[white!60!\blue] (2,2) circle(.6mm);
  \fill[white!60!\green] (.1,2.1) -- (.9,2.1) -- (.9,2.6) -- (.1,2.6);
  \fill[white!60!\green] (1.1,2.1) -- (1.9,2.1) -- (1.9,2.6) -- (1.1,2.6);
  \fill[white!60!\green] (2.1,.1) -- (2.1,.9) -- (2.6,.9) -- (2.6,.1);
  \fill[white!60!\green] (2.1,1.1) -- (2.1,1.9) -- (2.6,1.9) -- (2.6,1.1);
  \fill[white!60!\green] (2.1,2.1) -- (2.1,2.6) -- (2.6,2.6) -- (2.6,2.1);
  \everymath{\scriptscriptstyle}
  \draw[latex'-latex'] (0,2.6) node[above left=-3pt] {$y$} -- 
    (0,0) node[below] {$0$} -- (2.7,0) node[below right=-3pt] {$x$};
  \foreach \x in {1,2} 
    {\draw[dotted] (\x,0) node[below] {$\x$} -- +(0,2.6);
     \draw[dotted] (0,\x) node[left] {$\x$} -- +(2.6,0);}
  \path(1.5,3.1) node[draw,circle,inner sep=0pt] {$\scriptstyle #1$};
}
\begin{center}
\begin{minipage}[b]{.84\textwidth}
\centering
\begin{tikzpicture}[scale=.8]
\begin{scope}[xshift=-2.5cm,scale=.5]
  \drawregions{\ell_1}
  \draw[red!50,draw,line width=1mm] (0.1,0) -- (0.9,0);
  \draw[-latex'] (3.1,1.3) -- +(1,0) node[midway,above] {$\scriptstyle\mathsf{delay}$};
\end{scope}
\begin{scope}[scale=.5]
  \drawregions{\ell_1}
  \fill[red!50,line width=.5mm,draw] (1.2,1.1) -- (1.9,1.1) -- (1.9,1.8) -- cycle;
  \draw[-latex'] (3.1,1.3) -- +(1,0) node[midway,above] {$\scriptstyle x< 2$} node[midway,below] {$\scriptstyle x:=0$};
\end{scope}
\begin{scope}[xshift=2.5cm,scale=.5]
  \drawregions{\ell_2}
  \draw[blue!50,draw,line width=1mm] (0,1.1) -- (0,1.9);
  \draw[-latex'] (3.1,1.3) -- +(1,0) node[midway,above] {$\scriptstyle\mathsf{delay}$};
\end{scope}
\begin{scope}[xshift=5cm,scale=.5]
  \drawregions{\ell_2}
  \draw[blue!50,draw,line width=1mm] (0.1,2) -- (0.9,2);
  \draw[-latex'] (3.1,1.3) -- +(1,0) node[midway,above] {$\scriptstyle y= 2$} node[midway,below] {$\scriptstyle y:=0$};
\end{scope}
\begin{scope}[xshift=7.5cm,scale=.5]
  \drawregions{\ell_1}
  \draw[red!50,draw,line width=1mm] (0.1,0) -- (0.9,0);
\end{scope}
\begin{scope}[yshift=-2cm,line width=1pt]
\begin{scope}[xshift=-2.5cm,scale=1]
\draw (0,0) node[circle,inner sep=1pt,fill=black] (A0) {}; 
\draw (1,0) node[circle,inner sep=1pt,fill=black] (B0) {}; 
\draw[dotted] (A0) -- (B0);
\end{scope}
\begin{scope}[scale=1]
\fill[gray!30] (0,0) -- (1,0) -- (1,1) -- cycle;
\draw (0,0) node[circle,inner sep=1pt,fill=black] (A1) {}; 
\draw (1,0) node[circle,inner sep=1pt,fill=black] (B1) {}; 
\draw (1,1) node[circle,inner sep=1pt,fill=black] (C1) {}; 
\draw[dotted] (A1) -- (B1) -- (C1) -- (A1);
\end{scope}
\begin{scope}[xshift=2.5cm,scale=1]
\draw (0,0) node[circle,inner sep=1pt,fill=black] (A2) {}; 
\draw (0,1) node[circle,inner sep=1pt,fill=black] (B2) {}; 
\draw[dotted] (A2) -- (B2);
\end{scope}
\begin{scope}[xshift=5cm,scale=1]
\draw (0,1) node[circle,inner sep=1pt,fill=black] (A3) {}; 
\draw (1,1) node[circle,inner sep=1pt,fill=black] (B3) {}; 
\draw[dotted] (A3) -- (B3);
\end{scope}
\begin{scope}[xshift=7.5cm,scale=1]
\draw (0,0) node[circle,inner sep=1pt,fill=black] (A4) {}; 
\draw (1,0) node[circle,inner sep=1pt,fill=black] (B4) {}; 
\draw[dotted] (A4) -- (B4);
\end{scope}
\begin{scope}[shorten >=1mm,shorten <=1mm,color=black!60!white,line width=.5pt]
\draw[-latex'] (A0) .. controls +(80:1.7cm) and +(165:1.7cm) .. (C1);
\draw[-latex'] (A0) .. controls +(-25:1cm) and +(-155:1cm) .. (A1);
\draw[-latex'] (B0) .. controls +(-25:1cm) and +(-155:1cm) .. (B1);
\draw[-latex'] (A1) .. controls +(-25:1cm) and +(-155:1cm) .. (A2);
\draw[-latex'] (C1) .. controls +(-15:.7cm) and +(195:.7cm) .. (B2);
\draw[-latex'] (B1) .. controls +(15:.7cm) and +(165:.7cm) .. (A2);
\draw[-latex'] (B2) .. controls +(-15:.7cm) and +(195:.7cm) .. (A3);
\draw[-latex'] (A2) .. controls +(15:1cm) and +(-100:1.5cm) .. (B3);
\draw[-latex'] (A3) .. controls +(30:1cm) and +(100:.7cm) .. (A4);
\draw[-latex'] (B3) .. controls +(-80:.7cm) and +(-150:1cm) .. (B4);
\end{scope}
\end{scope}
\end{tikzpicture}
\end{minipage}
\begin{minipage}[b]{.14\textwidth}
 \centering
 \begin{tikzpicture}[line width=1pt]
    \begin{scope}
		\fill[gray!30] (0,0) -- (1,0) -- cycle;
		\draw (0,0) node[circle,inner sep=1pt,fill=black] (A1) {}; 
		\draw (1,0) node[circle,inner sep=1pt,fill=black] (B1) {}; 
        \node[] at (0.0, 0.2)   (c0) {$c_0$};
        \node[] at (1.0, 0.2)   (c1) {$c_1$};  
		\draw[dotted,use as bounding box] (A1) -- (B1) -- (A1);
		\begin{scope}[shorten >=1mm,shorten <=1mm,color=black,line width=.5pt]
		\draw[-latex'] (A1) .. controls +(180:6mm) and +(225:6mm) .. (A1);
		\draw[-latex'] (B1) .. controls +(-67.5:6mm) and +(-22.5:6mm) .. (B1);
		\draw[-latex'] (A1) .. controls +(-22.5:4mm) and +(-157.5:4mm) .. (B1);
  
		\end{scope}
	\end{scope}
	\end{tikzpicture}
 \vspace{15mm}
 \end{minipage}
\end{center}
\vspace{-0.4cm}
\caption{The orbit graph of a region cycle in the automaton of \figurename~\ref{fig:taEx}, and its folding.}
\label{fig:orbit-graph}
\vspace{-0.5cm}
\end{figure}

\paragraph{Folded orbit graphs}
Let $\pi$ be a region cycle around a region $r$, \textit{i.e.}~a finite region path that starts and ends in the same region state $(\ell,r)$. %
let $\mathcal C$ be the corners of $r$. 
The folded orbit graph of $\pi$ is a graph $\Gamma(\pi)$ of vertices $\mathcal C$. For every corner $c$ of $r$ and every corner $c'$ of $r$,
there is an edge from $c$ to $c'$ in $\Gamma(\pi)$ if there is a path from $(\mathsf{start},c)$ to $(\mathsf{end},c')$ in the orbit graph $\gamma(\pi)$.

It follows by the orbit graph construction that $\Gamma(\pi)$ represents the reachability relation from corner to corner along $\pi$.

\paragraph{Graph terminology}
We will use classical directed graph theory to describe the properties of folded orbit graphs:
A graph is said to be strongly connected if every vertex is reachable from every other vertex. The strongly connected components (SCCs) of a graph form a partition into subgraphs that are strongly connected.
Likewise, a graph is said to be weakly connected if its underlying undirected graph is strongly connected. The weakly connected components of a graph also form a partition into subgraphs that are weakly connected.
The disjoint union of graphs is constructed by making the vertex set of the result be the disjoint union of the vertex sets of the given graphs, and by making the edge set of the result be the disjoint union of the edge sets of the given graphs.
A graph is complete if every pair of vertices is connected by an edge.\footnote{including self-loops $v\rightarrow v$ for every vertex $v$.}

Finally, a \emph{cluster graph} is a disjoint union of complete graphs,
\textit{i.e.}~a graph where every weakly connected component is a complete SCC.

\section{Robustness of a finite path}
\label{sec:robust-path}

We define robust paths, as a notion meant to represent finite region paths that $\cont$ can traverse no matter what $\pert$ does:
\begin{definition}%
    We say that a well-formed finite region path $\pi$ is \emph{robust} if 
    every edge transition $(\ell,r)\xrightarrow{g,R}(\ell',r')$ in $\pi$ with a non-punctual guard $g$ satisfies that $r$ is non-punctual as well.
\end{definition}
Intuitively, %
this prevents situations where any non-zero perturbation would force the run out of $\pi$: the region $r$ is reached right after a delay transition that includes a perturbation in this case, so it cannot be punctual if we wish for $\cont$ to guarantee reaching it.
Note that the concatenation of robust paths is robust.

In order to show that this notion characterize robustness, we recall classical data structures and the notion of controllable predecessors of a set of valuations.

\paragraph{Zones and Difference Bound Matrices}
A (rational) Difference Bound Matrix (DBM) over a set of clocks $\Clocks=\{x_1,\dots x_n\}$ is a matrix $(M_{i,j})_{1\leq,i,j\leq n}$ so that each entry $M_{i,j}$ is a pair 
$(\preceq_{i,j}, m_{i,j})$ with $\preceq_{i,j}\in\{\leq,<\}$ and $m_{i,j}\in\Rat\cup\{+\infty\}$.
It represents a set of valuations $z$ called a zone, so that $\nu\in z$ if and only if
for each $1\leq i,j\leq n$, $\nu\models x_i-x_j \preceq_{i,j} m_{i,j}$.
These matrices can be used as efficient ways to represent guards or regions,
and many useful operations over sets of valuations can be described as polynomial matrix operations on DBMs.
\begin{definition}[\cite{bengtsson2003timed}]\label{def:dbm-operations}
    If $z$ is a zone, %
    $g$ is a guard, $R$ is a set of clocks and $\ClockBound$ is a bound on clocks,
    we define the following operations for any rational $\delta\geq 0$:
    \begin{itemize}
        \item $\mathsf{PreTime}_{\geq \delta}(z)=\{\nu\in\Valuations\mid \exists d\geq \delta, \nu+d\in z\}$,
        \item $g\cap z=\{\nu\in z\mid \nu\models g\}$,
        \item $\mathsf{Unreset}_{R}(z)=\{\nu\in\Valuations\mid \nu[R:=0]\in z\}$, and
        \item $\mathsf{Shrink}_{\delta}(z)=\{\nu\in\Valuations\mid \forall \varepsilon\in[-\delta,+\delta],\nu+\varepsilon\in z\}$.
    \end{itemize}
\end{definition}
In fact, if $z$ is encoded as a DBM, the output of the these operations is a DBM where entries are of the shape $(\preceq, m-\delta p)$ with $p\in\Nat$.

In order to make symbolic computations that abstract the perturbation value, one can use shrunk DBMs \cite{SankurBM14}, a parametric DBM where $\delta$ is a parameter.
Shrunk DBMs are pairs $(M,P)$ where $M$ is a DBM and $P$ is a shrinking matrix of the shape $(p_{i,j})_{1\leq,i,j\leq n}$ so that every entry $p_{i,j}$ is in $\Nat$.
Then, for any value of $\delta$, $M-\delta P$ represents the DBM of entry $(\preceq_{i,j},m_{i,j}-\delta p_{i,j})$.
The DBM operations of Definition~\ref{def:dbm-operations} can all be implemented on shrunk DBMs as parametric computations with an arbitrarily small $\delta>0$.

\paragraph{Controllable predecessors} 
We recall the definition of the $\CPre$ operation:
\begin{definition}[\cite{Sankur13}]
    Let $\pi=(\ell,r)\xrightarrow{\mathsf{delay}}(\ell,r')\xrightarrow{g,R}(\ell',r'')$ be an atomic robust region path (with $r'\models g$), let $\delta\geq 0$ be some amplitude of perturbation and $s$ be a set of valuations in $r''$.
    If $g$ is a non-punctual guard, then we let 
    \[\CPre_\pi^\delta(s)=\{\nu \in r \mid \exists d\geq \delta, \forall d'\in[d-\delta,d+\delta], \nu+d'\in r'\land (\nu+d')[R:=0]\in s\}\,.\]
    If $g$ is a punctual guard, then we let %
     \[\CPre_\pi^\delta(s)=\CPre_\pi^0(s)=\{\nu \in r \mid \exists d\geq 0, \nu+d\in r' \land (\nu+d)[R:=0]\in s\}\,.\]
    If $\pi=\pi'\pi''$, %
    then $\CPre_\pi^\delta(s)$ is defined inductively as $\CPre_{\pi'}^\delta(\CPre_{\pi''}^\delta(s))$.
\end{definition}
Overall, $\nu\in\CPre_{\pi}^\delta(s)$ if and only if Controller has a strategy that ensure reaching $s$ when starting from $\nu$ and following the path $\pi$, against any Perturbator strategy of maximal perturbation $\delta$.
If in particular $\delta=0$, $\CPre_{\pi}^0(s)$ is the standard predecessor operator over region paths, \textit{i.e.}~the set of valuations $\nu$ so that there exists a well-formed run from $\nu$ to some $\nu'\in s$ that follows $\pi$ without perturbations.
Thus, if we set $s=r'$, the last region of $\pi$, it holds that $\CPre_{\pi}^0(r')$ equals $r$, the first region of $\pi$, because regions represent equivalence classes of the time-abstract bisimulation relation~\cite{AD-tcs94}. 
We note the following inclusion properties: %
if $\delta\leq\delta'$, then $\CPre_{\pi}^{\delta'}(s)\subseteq\CPre_{\pi}^{\delta}(s)$, and
if $s\subseteq s'$, then $\CPre_{\pi}^\delta(s)\subseteq\CPre_{\pi}^{\delta}(s')$.

Let $\pi$ be a region path from $r$ to $r'$, and let $z$ be a zone in $r'$ represented as a DBM (or a shrunk DBM), then $\CPre_\pi^\delta(z)$ can be computed as a DBM (or as a shrunk DBM):
if $\pi$ is an atomic robust path with a non-punctual guard, then $\CPre_\pi^\delta(z)=r\cap\mathsf{PreTime}_{\geq \delta}(\mathsf{Shrink}_{\delta}(g\cap \mathsf{Unreset}_{R}(z)))$,
and if $\pi$ is an atomic robust path with a punctual guard, then $\CPre_\pi^\delta(z)=r\cap\mathsf{PreTime}_{\geq 0}(g\cap \mathsf{Unreset}_{R}(z))$.
\begin{restatable}[]{proposition}{robpath}
\label{cor:robpath}
    Let $\pi$ be a well-formed region path from $(\ell,r)$ to $(\ell',r')$.
    Then, there exists $\delta>0$ so that $\CPre_\pi^\delta(r')\neq\emptyset$ if and only if $\pi$ is robust.
\end{restatable}

In particular, if $\pi$ is a robust region path that starts from $(\ell_0,r_0)$ and ends in $(\ell,r)$ with $r_0$ the region $\{\mathbf{0}\}$, then $\mathbf{0}\in\CPre_\pi^\delta(r)$, so that $\cont$ can guarantee reaching $(\ell,r)$ from $(\ell_0,\mathbf{0})$ for a small enough $\delta$.
As a result, the notion of robust region paths is sufficient to enforce a reachability objective.
Winning for a Büchi objective, however, requires enforcing an infinite path, that will stay within a lasso forever.
If $\pi$ is a robust region cycle around a region $r$, we know from the inclusion properties of $\CPre_{\pi}^\delta$ that for all $k\in\Natp$, $r\supseteq \CPre_{\pi}^\delta(r)\supseteq\CPre_{\pi^{2}}^\delta(r)\supseteq\dots\supseteq\CPre_{\pi^k}^\delta(r)$.
For any fixed $\delta>0$, this decreasing sequence of sets $\CPre_{\pi^k}^\delta(r)$ may become empty for $k$ big enough, meaning there is no strategy of $\cont$ that is able to iterate the cycle $\pi$ forever. 

We will show that for a specific class of cycles $\pi$ this scenario does not happen, as we will find some set $s$ in $r$ so that
$s\subseteq\CPre_{\pi}^\delta(s)\subseteq\CPre_{\pi^{2}}^\delta(s)\subseteq\dots$, so that $\CPre_{\pi^k}^\delta(r)\neq\emptyset$ for all $k\geq 1$, \textit{i.e.}~$\cont$ can iterate $\pi$ forever.

\section{Robustness of an infinite path}
\label{sec:partRobCharc}

As previously explained, for a cycle in the region abstraction, being a robust path does not imply that it can be iterated forever, so that solving for Büchi objectives requires a finer notion.
  We define robustly iterable cycles %
  to characterise cycles that $\cont$ can iterate forever, no matter what $\pert$ does:
\begin{definition}%
    We say that a well-formed region cycle $\pi$ is \emph{robustly iterable} if $\pi$ is a robust path and $\Gamma(\pi)$ is a cluster graph.
    We say that a lasso $\pi_0\pi^\omega$ is \emph{robustly iterable} if $\pi_0$ is either empty or a robust path, and if $\pi$ is robustly iterable.
\end{definition}

The cluster graph condition represents a requirement on the reachability relation along $\pi$ that generalises the notion of aperiodic cycle from~\cite{Sankur13,AB-formats11,DBLP:conf/formats/Stainer12}  (such cycles had an iterate with a complete $\Gamma(\pi)$, which is in particular a cluster graph).
This notion is stable through iterations of cycles, so that if $\pi$ is a robustly iterable cycle, then for every $l\geq 1$, the cycle $\pi^l$ has the same folded orbit graph as $\pi$ (by decomposition into concatenated orbit graphs $\gamma(\pi)\dots\gamma(\pi)$), and therefore is robustly iterable as well.
On the other hand, it is possible for a cycle $\pi$ that is not robustly iterable to become robustly iterable after iterating it multiple times.
Let us detail what can happen as a cycle is iterated:

\begin{restatable}[]{lemma}{robCycle}
    \label{lm:robCycle}
    Let $\pi$ be a region cycle around a region $r$ of dimension $m$, so that $\pi$ is a robust path.
    Then only one of the following cases must hold:
    \begin{itemize}
        \item Either there exists $1\leq k\leq m(m+1)!$ so that $\pi^k$ is robustly iterable,
        \item or there exists $1\leq k\leq (m+1)!$ so that $\Gamma(\pi^k)$ contains a weakly connected component that is not strongly connected.
    \end{itemize}
\end{restatable}

We will show that $\cont$ 
can ensure staying in $\pi$ forever
in the first case of Lemma~\ref{lm:robCycle}, %
but that in the second case and for any fixed $\delta>0$, $\pert$ can prevent him from staying in $\pi$ forever. %

\begin{theorem}%
\label{thm:main}
Given an instance of the robust controller synthesis problem of automaton $\TA$ and objective $\mathsf{Buchi}$,
there exists $\delta>0$ so that $\cont$ wins $\thegame$
if and only if 
there exists 
a winning lasso that is robustly iterable.
\end{theorem}

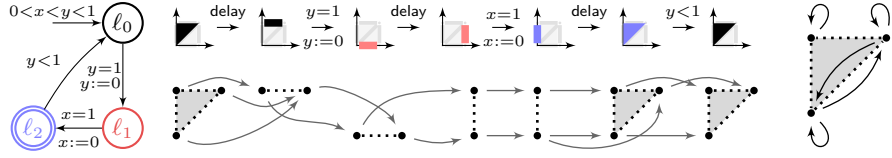
\begin{figure}[ht!]
  \begin{center}
\vspace*{10pt}

\def\yellow{black!10!white}
\def\red{black!30!white}
\def\green{black!10!white}
\def\blue{black!30!white}
\def\drawregions#1{%
  \fill[white!60!\yellow] (.2,.1) -- (.9,.1) -- (.9,.8) -- cycle;
  \fill[white!60!\yellow] (.1,.2) -- (.1,.9) -- (.8,.9) -- cycle;
  \draw[white!60!\red,line width=.5mm] (0.1,0) -- (0.9,0);
  \draw[white!60!\red,line width=.5mm] (0.1,1) -- (0.9,1);
  \draw[white!60!\red,line width=.5mm] (1,0.1) -- (1,0.9);
  \draw[white!60!\red,line width=.5mm] (0,0.1) -- (0,0.9);
  \draw[white!60!\red,line width=.5mm] (.1,.1) -- (.9,.9);
  \fill[white!60!\blue] (0,0) circle(.6mm);
  \fill[white!60!\blue] (1,0) circle(.6mm);
  \fill[white!60!\blue] (0,1) circle(.6mm);
  \fill[white!60!\blue] (1,1) circle(.6mm);
  \everymath{\scriptscriptstyle}
  \draw[latex'-latex'] (0,1.7)  -- 
    (0,0)  -- (1.7,0) ;
} 
  
\begin{tikzpicture}
    \begin{scope}[xshift=-1.85cm,yshift=0.3cm,node distance=2.5cm,auto, scale = .7]
    \tikzstyle{every state}=[thick, circle,minimum size=15pt,inner sep=0pt]
    \node (preA) at (-1.5,0) {};
    \node[state] at (0,0) (A) {$\ell_0$};
    \node[state,color =red!60!lightgray ] at (0,-2) (B) {$\ell_1$};
    \node[accepting,state,color =blue!50] at (-1.7,-2) (C) {$\ell_2$};
    \everymath{\scriptstyle}
    \path[-latex'] 
    (preA) edge node[above=-2pt,pos=0] {$0<x<y<1$} (A) %
    (A) edge node[left,xshift=3pt,pos=0.40] {$y=1$}
             node[left,xshift=2.3pt,pos=0.64] {$y:=0$}(B)
    (C) edge[bend left=10] node[above,pos=0.4,xshift=-0.3cm] {$y<1$} (A)
    (B) edge node[above] {$x=1$} node[below] {$x:=0$} (C);
    \end{scope}

    \begin{scope}[scale=.6]
\begin{scope}[xshift=-1.9cm,scale=.5]
  \drawregions{\ell_0}
  \fill[black,line width=.5mm,draw] (0.8,0.9) -- (0.1,0.2) -- (0.1,0.9) -- cycle;
  \draw[-latex'] (1.8,0.8) -- +(1,0) node[midway,above] {$\scriptstyle\mathsf{delay}$};
\end{scope}
\begin{scope}[scale=.5]
  \drawregions{\ell_0}
  \draw[black,draw,line width=1mm] (0.1,1.0) -- (0.9,1.0);
  \draw[-latex'] (2.2,0.8) -- +(1.0,0) node[midway,above] {$\scriptstyle y= 1$} node[midway,below] {$\scriptstyle y:=0$};
\end{scope}
\begin{scope}[xshift=2.1cm,scale=.5]
  \drawregions{\ell_1}
  \draw[red!50,draw,line width=1mm] (0.1,0.0) -- (0.9,0.0);
  \draw[-latex'] (1.8,0.8) -- +(1,0) node[midway,above] {$\scriptstyle\mathsf{delay}$};
\end{scope}
\begin{scope}[xshift=4.0cm,scale=.5]
  \drawregions{\ell_1}
  \draw[red!50,draw,line width=1mm] (1.0,0.1) -- (1.0,0.9);
  \draw[-latex'] (2.2,0.8) -- +(1,0) node[midway,above] {$\scriptstyle x= 1$} node[midway,below] {$\scriptstyle x:=0$};
\end{scope}
\begin{scope}[xshift=6.1cm,scale=.5]
  \drawregions{\ell_2}
  \draw[blue!50,draw,line width=1mm] (0.0,0.1) -- (0.0,0.9);
  \draw[-latex'] (1.8,0.8) -- +(1,0) node[midway,above] {$\scriptstyle\mathsf{delay}$};
\end{scope}
\begin{scope}[xshift=8.0cm,scale=.5]
  \drawregions{\ell_2}
  \fill[blue!50,line width=.5mm,draw] (0.8,0.9) -- (0.1,0.2) -- (0.1,0.9) -- cycle;
  \draw[-latex'] (2.2,0.8) -- +(1,0) node[pos=0.45,above] {$\scriptstyle y<1$};
\end{scope}
\begin{scope}[xshift=10.0cm,scale=.5]
  \drawregions{\ell_0}
  \fill[black,line width=.5mm,draw] (0.8,0.9) -- (0.1,0.2) -- (0.1,0.9) -- cycle;
\end{scope}
  
\begin{scope}[yshift=-2cm,line width=1pt]
\begin{scope}[xshift=-1.9cm,scale=1]
\fill[gray!30] (0,0) -- (0,1) -- (1,1) -- cycle;
\draw (0,0) node[circle,inner sep=1pt,fill=black] (A0) {}; 
\draw (0,1) node[circle,inner sep=1pt,fill=black] (B0) {};
\draw (1,1) node[circle,inner sep=1pt,fill=black] (C0) {};
\draw[dotted] (A0) -- (B0) -- (C0) -- (A0);
\end{scope}
\begin{scope}[scale=1]
\draw (0,1) node[circle,inner sep=1pt,fill=black] (A1) {}; 
\draw (1,1) node[circle,inner sep=1pt,fill=black] (B1) {}; 
\draw[dotted] (A1) -- (B1);
\end{scope}
\begin{scope}[xshift=2.1cm,scale=1]
\draw (0,0) node[circle,inner sep=1pt,fill=black] (A2) {}; 
\draw (1,0) node[circle,inner sep=1pt,fill=black] (B2) {}; 
\draw[dotted] (A2) -- (B2);
\end{scope}
\begin{scope}[xshift=3.7cm,scale=1]
\draw (1,0) node[circle,inner sep=1pt,fill=black] (A3) {}; 
\draw (1,1) node[circle,inner sep=1pt,fill=black] (B3) {}; 
\draw[dotted] (A3) -- (B3);
\end{scope}
\begin{scope}[xshift=6.1cm,scale=1]
\draw (0,0) node[circle,inner sep=1pt,fill=black] (A4) {}; 
\draw (0,1) node[circle,inner sep=1pt,fill=black] (B4) {}; 
\draw[dotted] (A4) -- (B4);
\end{scope}
\begin{scope}[xshift=7.8cm,scale=1]
    \fill[gray!30] (0,0) -- (0,1) -- (1,1) -- cycle;
\draw (0,0) node[circle,inner sep=1pt,fill=black] (A5) {}; 
\draw (0,1) node[circle,inner sep=1pt,fill=black] (B5) {};
\draw (1,1) node[circle,inner sep=1pt,fill=black] (C5) {};
\draw[dotted] (A5) -- (B5) -- (C5) -- (A5);
    \end{scope}
    \begin{scope}[xshift=9.9cm,scale=1]
\fill[gray!30] (0,0) -- (0,1) -- (1,1) -- cycle;
\draw (0,0) node[circle,inner sep=1pt,fill=black] (A6) {}; 
\draw (0,1) node[circle,inner sep=1pt,fill=black] (B6) {};
\draw (1,1) node[circle,inner sep=1pt,fill=black] (C6) {};
\draw[dotted] (A6) -- (B6) -- (C6) -- (A6);
\end{scope}

\begin{scope}[shorten >=1mm,shorten <=1mm,color=black!60!white,line width=.5pt]
\draw[-latex'] (B0) .. controls +(20:1.0cm) and +(160:1.0cm) .. (A1);
\draw[-latex'] (A0) .. controls +(-25:1cm) and +(-132.5:1cm) .. (B1);
\draw[-latex'] (C0) .. controls +(-25:1cm) and +(-155:1cm) .. (B1);
\draw[-latex'] (A1) .. controls +(-60:1cm) and +(155:1cm) .. (A2);
\draw[-latex'] (B1) .. controls +(15:.7cm) and +(120:.7cm) .. (B2);
\draw[-latex'] (B2) .. controls +(-15:.7cm) and +(195:.7cm) .. (A3);
\draw[-latex'] (A2) .. controls +(60:1cm) and +(-180:1.5cm) .. (B3);
\draw[-latex'] (A3) --  (A4);
\draw[-latex'] (B3) -- (B4);
\draw[-latex'] (A4) --  (A5);
\draw[-latex'] (B4) --  (B5);
\draw[-latex'] (A4) .. controls +(-22.5:1.5cm) and +(-90:1.0cm) .. (C5);

\draw[-latex'] (A5) --  (A6);
\draw[-latex'] (B5) .. controls +(20.0:1.0cm) and +(160:1.0cm) ..  (B6);
\draw[-latex'] (C5) .. controls +(20.0:1.0cm) and +(160:1.0cm) .. (C6);
\end{scope}
\end{scope}
\end{scope}
\begin{scope}[line width=1pt,xshift=7.3cm,yshift=-0.9cm]
            \fill[gray!30] (0,0) -- (1,1) -- (0,1) -- cycle;
            \draw (0,0) node[circle,inner sep=1pt,fill=black] (A1) {}; 
            \draw (1,1) node[circle,inner sep=1pt,fill=black] (B1) {}; 
            \draw (0,1) node[circle,inner sep=1pt,fill=black] (C1) {}; 
            \draw[dotted,use as bounding box] (A1) -- (B1) -- (C1) -- (A1);
            \begin{scope}[shorten >=1mm,shorten <=1mm,color=black,line width=.5pt]
            \draw[-latex'] (A1) .. controls +(-45:6mm) and +(-90:6mm) .. (A1);
            \draw[-latex'] (B1) .. controls +(80:6mm) and +(125:6mm) .. (B1);
            \draw[-latex'] (C1) .. controls +(45:6mm) and +(90:6mm) .. (C1);
            \draw[-latex'] (A1) .. controls +(22.5:6mm) and +(-112.5:6mm) .. (B1);
            \draw[-latex'] (B1) .. controls +(-157.5:6mm) and +(67.5:6mm) .. (A1);
            \end{scope}
        \end{scope}
        \end{tikzpicture}    
    
    \caption{\label{fig:cluster}A timed automaton with a robustly iterable cycle $\ell_0\ell_1\ell_2$ is on the left. 
    The corresponding orbit graph is displayed in the center and its folded orbit graph is depicted on the right.
    The folded orbit graph is a cluster graph with two complete SCCs.}
  \end{center}
  \vspace{-.5cm}
\end{figure}

This characterisation based on robustly iterable lassos can be naturally extended to a non-deterministic algorithm using polynomial space, which implies membership to $\PSPACE$ for our problem. The algorithm consists in guessing a winning lasso that is robustly iterable, computing its folded orbit graph and checking that it is indeed a cluster graph. The details are somewhat involved but classical for this kind of problems (see \textit{e.g.} the procedure in~\cite[Section~8.7]{Sankur13} for a different characterisation based on orbit graphs), as the region path needs to be guessed transition after transition, so that the folded orbit graph is built by composition in an online fashion, and does not require the entire cycle to be stored in memory. The length of the cycle that is guessed can be exponentially bounded by combinatorial reasoning (on the number of folded orbit graphs that can be built from a given timed automaton).
The $\pspace$-hardness of our problem straightforwardly follows from~\cite{SBMR-concur13} where the robust controller synthesis problem without punctual guards was shown to be $\pspace$-complete.
Hence:
\begin{theorem}
The robust controller synthesis problem is $\pspace$-complete.
\end{theorem}

The remainder of this paper is devoted to the proof of the correctness of our characterization, i.e.~Theorem~\ref{thm:main}.

\section{Slicing regions}
\label{sec:Slices}
Given a single region $r$, we define a partition of $r$ into sets of valuations called slices according to a partition of the corners of $r$.
\begin{definition}%
    Let $r$ be a region of dimension $m$ and corner $\mathcal{C}$. %
    Let $0\leq k\leq m$.
    A \emph{corner partition} of $r$ is a partition $C_0\uplus  \dots\uplus C_k$ of $\mathcal{C}$
    into $k+1$ colors, such that a corner $c\in\mathcal C$ is said to have color $0\leq i\leq k$ if $c\in C_i$.
\end{definition}

Recall that the folded orbit graph of a robustly iterable region cycle is a cluster graph, \textit{i.e.}~a disjoint union of complete graphs.

\begin{definition}%
    Let $\pi$ be a robustly iterable region cycle around a region $r$, of folded orbit graph $\Gamma(\pi)$.
    The \emph{cluster corner partition} of $r$ induced by $\Gamma(\pi)$ is the corner partition of $r$ 
    so that every color corresponds to a complete SCC of $\Gamma(\pi)$.
\end{definition}
In other words, corners $c$ and $c'$ have the same color in the cluster corner partition if and only if they are in the same complete SCC %
of $\Gamma(\pi)$.

\begin{definition}%
    Let $C_0\uplus \dots\uplus C_k$ be a corner partition of a region $r$ of dimension $m$ into $k+1$ colors, ordered from $0$ to $k$ w.l.o.g..
    A \emph{color weight vector} is a weight vector $\vecw=(w_0,\dots, w_k)\in(0,1]^{k+1}$ so that $\|\vecw\|_1=\sum_{i=0}^k w_i=1$.
    In this context, we say that $\vecw$ gives color $0\leq i\leq k$ the weight $w_i$.
\end{definition}

\begin{definition}[Slice partition]
    Let $C_0\uplus \dots\uplus C_k$ be a corner partition of a region $r$ of dimension $m$ into $k+1$ colors.
    A \emph{slice} w.r.t. $C_0\uplus \dots\uplus C_k$ is a set of valuations defined by a color weight vector $\vecw=(w_0,\dots, w_k)$.
    Let $\nu$ be a valuation in $r$ of corner weights $\vecl=(\lambda_0,\dots,\lambda_m)$.
    Then, $\nu$ belongs to the slice defined by $\vecw$, denoted $\slice^r_{C_0\uplus \dots\uplus C_k}(\vecw)$, if for all $0\leq j\leq k$, $\sum\limits_{\substack{0\leq i\leq m\\\text{s.t. }c_i\in C_j}} \lambda_i = w_j$.
\end{definition}

    Whenever the partition $C_0\uplus \dots\uplus C_k$ is the cluster corner partition of a region $r$ induced by $\Gamma(\pi)$,
    where $\pi$ is a robustly iterable region cycle around $r$,
    we write $\slice_\pi(\vecw)$ instead of $\slice^r_{C_0\uplus \dots\uplus C_k}(\vecw)$.

\begin{remark}
Intuitively, $\vecw$ gives a total weight for each color, such that the slice defined by $\vecw$ contains the convex combinations of corners compatible with these total weights when summing every corner weight of the same color.
As $\vecw$ ranges over every color weight vector, the set $\{\slice_\pi(\vecw)\}$ partitions the entire region, so that each valuation in the region belongs to a unique slice.
Slices are convex polyhedra, but may not be zones.\footnote{On automata with $3$ clocks, some corner partitions may induce slices described by linear equations of the shape $x+y-z=w$.}
However, %
slices defined w.r.t. the cluster corner partition of a robustly iterable cycle are always zones~(cf.~Appendix~\ref{app:slices}).
\end{remark}

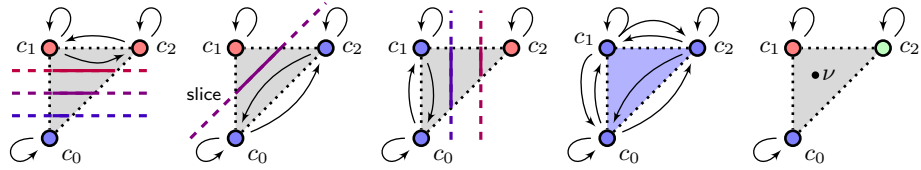
\begin{figure}[t!]
  \begin{center}
  \vspace{2mm}
  \begin{minipage}[c]{.19\textwidth}
  \centering
  \begin{tikzpicture}[line width=1pt, scale = 1.2]
    \useasboundingbox (-.1,-.1) rectangle (1.1,1.1);
    \begin{scope}
        \fill[gray!30] (0,0) -- (0,1) -- (1,1) -- cycle; 
		\draw (0,0) node[circle,inner sep=2pt,draw=black,fill=blue!50] (A1) {}; 
        \node[] at (.25, -.2)   (c0) {$c_0$};
		\draw (1,1) node[circle,inner sep=2pt,draw=black,fill=red!50] (B1) {}; 
        \node[] at (-.25, 1)   (c1) {$c_1$};  
		\draw (0,1) node[circle,inner sep=2pt,draw=black,fill=red!50] (C1) {}; 
        \node[] at (1.3, 1)   (c2) {$c_2$};  
		\draw[dotted,use as bounding box] (A1) -- (B1) -- (C1) -- (A1);
		\begin{scope}[shorten >=1mm,shorten <=1mm,color=black,line width=.5pt]
		\draw[-latex'] (A1) .. controls +(180:6mm) and +(225:6mm) .. (A1);
		\draw[-latex'] (B1) .. controls +(45:6mm) and +(90:6mm) .. (B1);
		\draw[-latex'] (C1) .. controls +(45:6mm) and +(90:6mm) .. (C1);
		\draw[-latex'] (C1) .. controls +(-22.5:6mm) and +(-157.5:4mm) .. (B1);
        \draw[-latex'] (B1) .. controls +(157.5:4mm) and +(22.5:4mm) .. (C1);
        \draw[use as bounding box, very thick,red!50!blue] (-.05,.5) -- (.55, .5);
        \draw[dashed,use as bounding box, very thick,red!50!blue] (-.5,.5) -- (1.2, .5);
        
        \draw[use as bounding box, very thick,red!25!blue] (-.05,.25) -- (.3, .25);
        \draw[dashed,use as bounding box, very thick,red!25!blue] (-.5,.25) -- (1.2, .25);
        
        \draw[use as bounding box, very thick,red!75!blue] (-.05,.75) -- (.8, .75);
        \draw[dashed,use as bounding box, very thick,red!75!blue] (-.5,.75) -- (1.2, .75);
		\end{scope}
	\end{scope}
	\end{tikzpicture}  
 \end{minipage}
 \hfill
    \begin{minipage}[c]{.19\textwidth}
    \centering
  \begin{tikzpicture}[line width=1pt, scale = 1.2]
    \useasboundingbox (-.1,-.1) rectangle (1.1,1.1);
    \begin{scope}
        \fill[gray!30] (0,0) -- (0,1) -- (1,1) -- cycle;
        \draw[use as bounding box, very thick,red!50!blue] (0,.5) -- (.5, 1);
        \draw[dashed,use as bounding box, very thick,red!50!blue] (-.5,0) -- (1, 1.5);
        \node[] at (-0.35, 0.5)   (c2) {\scriptsize $\slice$}; 
		\draw (0,0) node[circle,inner sep=2pt,draw=black,fill=blue!50] (A1) {}; 
        \node[] at (.25, -.2)   (c0) {$c_0$};
		\draw (1,1) node[circle,inner sep=2pt,draw=black,fill=blue!50] (B1) {}; 
        \node[] at (-.25, 1)   (c1) {$c_1$};
		\draw (0,1) node[circle,inner sep=2pt,draw=black,fill=red!50] (C1) {}; 
        \node[] at (1.3, 1)   (c2) {$c_2$};  
		\draw[dotted,use as bounding box] (A1) -- (B1) -- (C1) -- (A1);
		\begin{scope}[shorten >=1mm,shorten <=1mm,color=black,line width=.5pt]
		\draw[-latex'] (A1) .. controls +(180:6mm) and +(225:6mm) .. (A1);
		\draw[-latex'] (B1) .. controls +(45:6mm) and +(90:6mm) .. (B1);
		\draw[-latex'] (C1) .. controls +(45:6mm) and +(90:6mm) .. (C1);
		\draw[-latex'] (A1) .. controls +(22.5:6mm) and +(-112.5:6mm) .. (B1);
        \draw[-latex'] (B1) .. controls +(-157.5:6mm) and +(67.5:6mm) .. (A1);
		\end{scope}
	\end{scope}
	\end{tikzpicture}  
 \end{minipage}
 \hfill
    \begin{minipage}[c]{.19\textwidth}
    \centering
   \begin{tikzpicture}[line width=1pt, scale = 1.2]
    \useasboundingbox (-.1,-.1) rectangle (1.1,1.1);
    \begin{scope}
        \fill[gray!30] (0,0) -- (1,1) -- (0,1) -- cycle;
		\draw (0,0) node[circle,inner sep=2pt,draw=black,fill=blue!50] (A1) {}; 
        \node[] at (.25, -.2)   (c0) {$c_0$};
		\draw (1,1) node[circle,inner sep=2pt,draw=black,fill=red!50] (B1) {}; 
        \node[] at (-.25, 1)   (c1) {$c_1$};  
		\draw (0,1) node[circle,inner sep=2pt,draw=black,fill=blue!50] (C1) {}; 
        \node[] at (1.3, 1)   (c2) {$c_2$};  
		\draw[dotted,use as bounding box] (A1) -- (B1) -- (C1) -- (A1);
		\begin{scope}[shorten >=1mm,shorten <=1mm,color=black,line width=.5pt]
		\draw[-latex'] (A1) .. controls +(180:6mm) and +(225:6mm) .. (A1);
		\draw[-latex'] (B1) .. controls +(45:6mm) and +(90:6mm) .. (B1);
		\draw[-latex'] (C1) .. controls +(45:6mm) and +(90:6mm) .. (C1);
		\draw[-latex'] (A1) .. controls +(112.5:4mm) and +(-112.5:4mm) .. (C1);
        \draw[-latex'] (C1) .. controls +(-67.5:4mm) and +(67.5:4mm) .. (A1);
        \draw[use as bounding box, very thick,red!33!blue] (.33,.28) -- (.33, 1.1);
        \draw[dashed,use as bounding box, very thick,red!33!blue] (.33,-.1) -- (.33, 1.5);

        \draw[use as bounding box, very thick,red!66!blue] (.66,.61) -- (.66, 1.1);
        \draw[dashed,use as bounding box, very thick,red!66!blue] (.66,-.1) -- (.66, 1.5);
		\end{scope}
	\end{scope}
	\end{tikzpicture}    
 \end{minipage}
 \hfill
    \begin{minipage}[c]{.19\textwidth}
    \centering
  \begin{tikzpicture}[line width=1pt, scale = 1.2]
    \useasboundingbox (-.1,-.1) rectangle (1.1,1.1);
    \begin{scope}
        \fill[blue!30] (0,0) -- (0,1) -- (1,1) -- cycle;
		\draw (0,0) node[circle,inner sep=2pt,draw=black,fill=blue!50] (A1) {}; 
        \node[] at (.25, -.2)   (c0) {$c_0$};
		\draw (1,1) node[circle,inner sep=2pt,draw=black,fill=blue!50] (B1) {}; 
        \node[] at (-.25, 1.1)   (c1) {$c_1$};
		\draw (0,1) node[circle,inner sep=2pt,draw=black,fill=blue!50] (C1) {}; 
        \node[] at (1.3, 1)   (c2) {$c_2$};  
		\draw[dotted,use as bounding box] (A1) -- (B1) -- (C1) -- (A1);
		\begin{scope}[shorten >=1mm,shorten <=1mm,color=black,line width=.5pt]
		\draw[-latex'] (A1) .. controls +(180:6mm) and +(225:6mm) .. (A1);
		\draw[-latex'] (B1) .. controls +(45:6mm) and +(90:6mm) .. (B1);
		\draw[-latex'] (C1) .. controls +(45:6mm) and +(90:6mm) .. (C1);
		\draw[-latex'] (A1) .. controls +(22.5:6mm) and +(-112.5:6mm) .. (B1);
        \draw[-latex'] (B1) .. controls +(-157.5:6mm) and +(67.5:6mm) .. (A1);
		\draw[-latex'] (C1) .. controls +(45:6mm) and +(135:4mm) .. (B1);
        \draw[-latex'] (B1) .. controls +(157.5:4mm) and +(22.5:4mm) .. (C1);
		\draw[-latex'] (A1) .. controls +(112.5:4mm) and +(-112.5:4mm) .. (C1);
        \draw[-latex'] (C1) .. controls +(-157.5:4mm) and +(135:4mm) .. (A1);
		\end{scope}
	\end{scope}
	\end{tikzpicture}  
 \end{minipage}
 \hfill
    \begin{minipage}[c]{.19\textwidth}
    \centering
  \begin{tikzpicture}[line width=1pt, scale = 1.2]
    \useasboundingbox (-.1,-.1) rectangle (1.1,1.1);
    \begin{scope}
        \fill[gray!30] (0,0) -- (0,1) -- (1,1) -- cycle;
        \draw (0.25,0.7) node[circle,inner sep=1pt,fill=black] {};
        \node[] at (0.4, 0.7) {$\nu$};  
		\draw (0,0) node[circle,inner sep=2pt,draw=black,fill=blue!50] (A1) {}; 
        \node[] at (.25, -.2)   (c0) {$c_0$};
		\draw (1,1) node[circle,inner sep=2pt,draw=black,fill=green!25] (B1) {}; 
        \node[] at (-.25, 1)   (c1) {$c_1$};
		\draw (0,1) node[circle,inner sep=2pt,draw=black,fill=red!50] (C1) {}; 
        \node[] at (1.3, 1)   (c2) {$c_2$};  
		\draw[dotted,use as bounding box] (A1) -- (B1) -- (C1) -- (A1);
		\begin{scope}[shorten >=1mm,shorten <=1mm,color=black,line width=.5pt]
		\draw[-latex'] (A1) .. controls +(180:6mm) and +(225:6mm) .. (A1);
		\draw[-latex'] (B1) .. controls +(45:6mm) and +(90:6mm) .. (B1);
		\draw[-latex'] (C1) .. controls +(45:6mm) and +(90:6mm) .. (C1);
		\end{scope}
	\end{scope}
	\end{tikzpicture}  
 \end{minipage}
    \caption{Cluster corner partitions and some slices of a region of dimension 2.
    }
    \label{fig:slice}
  \end{center}
  \vspace{-.75cm}
\end{figure}

\begin{example}
    Let $r$ be the region $0<x<y<1$ of dimension $2$ of corners $c_0$, $c_1$ and $c_2$.
    \figurename~\ref{fig:slice} depicts
    three examples of corner partitions of $r$ into $2$ colors $C_0\uplus C_1$, %
    and examples of folded orbit graphs that induce these partitions as cluster corner partitions. A few slices are depicted in each case for different color weight vectors.
    In particular, the partition displayed in the second example in \figurename~\ref{fig:slice} is the cluster corner partition induced by the cycle $\pi$ from \figurename~\ref{fig:cluster}, and the slice drawn in this example is the slice $\slice_\pi(\vecw)$ for $\vecw=(\frac{1}{2},\frac{1}{2})$. %
    The two rightmost examples represent partitions with $1$ and $3$ colors that result in one slice equal to $r$, or singleton slices that contain a single valuation $\nu$, respectively.
\end{example}

As detailed in Appendix~\ref{app:slice-based-proof}, the slices induced by a robustly iterable cycle $\pi$ represent equivalence classes of the reachability relation, so that there is a run from $\nu$ to $\nu'$ following $\pi$ if and only if $\nu$ and $\nu'$ belong to the same slice.

\section{Proof of Theorem~\ref{thm:main}}
\label{sec:proofMainThm}

We now proceed with the proof of our characterization in term of cluster graphs. This is formalized in the following propositions.

\begin{restatable}[]{proposition}{contWins}\label{prop:contWins}
If there exists %
a winning lasso that is robustly iterable, then $\cont$ wins $\thegame$ for some $\delta>0$.
\end{restatable}
\begin{proof}[Sketch]
Recall that, as described in the end of Section~\ref{sec:robust-path}, in order to win, $\cont$ has to exhibit a cycle that can be iterated despite the perturbation. We will actually show that a \emph{robustly iterable lasso} can be iterated forever against any perturbation inflicted by $\pert$.
The full proof can be found in Appendix~\ref{app:slice-based-proof}, %
here we sketch the essential steps of this proof and explain how we use the different notions introduced so far.

Let $\pi_0\pi^\omega$ be a robustly iterable winning lasso around the region $r$. %
In \figurename~\ref{fig:contWins}, we depict the essential steps of our proof from left to right.
The first step of the proof, 
presented in \figurename~\ref{fig:contWins}.1,
consists in showing that %
one can build a DBM $N$ such that any shrinking of $r$ contains $N$, i.e., for any shrinking matrix $P$ and $\delta$ small enough, $N \subseteq r - \delta P$.
The construction of this DBM is detailed in Appendix~\ref{app:slice-based-proof}.

The second step of the proof (\figurename~\ref{fig:contWins}.2) consists in showing that
for any slice $\mathsf{slice}$ induced by the cluster corner partition that intersects $N$,
there exists a \emph{shrinking matrix} $P$ such that 
$(r - \delta P) \cap \mathsf{slice} = \CPre^\delta_\pi(N \cap \mathsf{slice})$.
This is established in Lemma~\ref{lm:sliceShrink} of Appendix~\ref{app:slice-based-proof}. To obtain this latter equation we use two crucial properties of any slice induced by the cluster corner partition: 
\begin{itemize}
    \item The fact that the folded orbit graph of $\pi$ is a cluster graph implies that the reachability relation along $\pi$ is \emph{complete} over the slice. 
    \item Slices induced by a folded orbit graph are always \emph{zones}, cf. Proposition~\ref{prop:slicezones} of Appendix~\ref{app:slices}, and are thus compatible with (shrunk) DBM operations.
\end{itemize}
Finally, the first two step
entail that, as displayed 
in \figurename~\ref{fig:contWins}.3:
\[N \cap \mathsf{slice}\subseteq \CPre^\delta_\pi(N \cap \mathsf{slice})\,.\]
This result implies that $\cont$ wins from any valuation $\nu$ in
$N \cap \mathsf{slice}$, %
as he can always enforce staying within $N \cap \mathsf{slice}$.
\end{proof}

\begin{figure}[ht!]
  \begin{center}
  \begin{minipage}[b]{.3\textwidth}
  \centering
  \begin{tikzpicture}[line width=1pt,scale=2]
    \useasboundingbox (-.1,-.1) rectangle (1.1,1.1);
    \begin{scope}
        \draw[dashed,fill=gray!20,thin] (0.05,.1) -- (0.05,0.95) -- (.9,.95) -- cycle;
        \draw[fill=gray!50,thin] (0.15,.3) -- (0.15,0.85) -- (.7,.85) -- cycle;
        \node[] at (-0.1,0.5) {\scriptsize $r$};
        \node[] at (0.25,0.75) {\scriptsize $N$};
        \node[] at (0.75,0.25) {\scriptsize $r-\delta P$};
        \draw [-latex',thin] (0.5,0.25) to [out=-167.5,in=-12.5] (0.1,0.25);
        
        \node[] at (0.5,1.3) {\scriptsize $N\subseteq r-\delta P$};
        \node[] at (0.5,-0.25) {(1)};
		\draw (0,0) node[circle,inner sep=2pt,draw=black,fill=blue!50] (A1) {}; 
		\draw (1,1) node[circle,inner sep=2pt,draw=black,fill=blue!50] (B1) {}; 
		\draw (0,1) node[circle,inner sep=2pt,draw=black,fill=red!50] (C1) {}; 
		\draw[dotted,use as bounding box] (A1) -- (B1) -- (C1) -- (A1);
	\end{scope}
	\end{tikzpicture}
 \end{minipage}
 \hfill
 \begin{minipage}[b]{.3\textwidth}
 \centering
  \begin{tikzpicture}[line width=1pt,scale=2]
    \useasboundingbox (-.1,-.1) rectangle (1.1,1.1);
    \begin{scope}
        \draw[dashed,fill=gray!30,thin] (0.05,.1) -- (0.05,0.95) -- (.9,.95) -- cycle;
        \draw[fill=gray!50,thin] (0.15,.3) -- (0.15,0.85) -- (.7,.85) -- cycle;
        \node[] at (-0.25,0.3) {\scriptsize \textcolor{black}{$\slice$}};
        \node[] at (0.6,0.25) {\scriptsize \textcolor{black}{$N\cap\slice$}};
        \draw [-latex',thin] (0.35,0.25) to [out=170,in=-90,looseness=1] (0.25,0.53);
        
        \node[] at (0.5,1.3) {\scriptsize $(r - \delta P) \cap \mathsf{slice} = \CPre^\delta_\pi(N \cap \mathsf{slice})$};
        \node[] at (0.5,-0.25) {(2)};
        
        \draw[use as bounding box, very thick,red!33!blue] (0.15,.48) -- (0.52, 0.85);
        \draw[dashed,use as bounding box, thin,red!33!blue] (-0.33,0) -- (0.83, 1.16);
        
		\draw (0,0) node[circle,inner sep=2pt,draw=black,fill=blue!50] (A1) {}; 
		\draw (1,1) node[circle,inner sep=2pt,draw=black,fill=blue!50] (B1) {}; 
		\draw (0,1) node[circle,inner sep=2pt,draw=black,fill=red!50] (C1) {}; 
		\draw[dotted,use as bounding box] (A1) -- (B1) -- (C1) -- (A1);
	\end{scope}
	\end{tikzpicture}   
 \end{minipage}
 \hfill
 \begin{minipage}[b]{.3\textwidth}
 \centering
  \begin{tikzpicture}[line width=1pt,scale=2]
    \useasboundingbox (-.1,-.1) rectangle (1.1,1.1);
    \begin{scope}
        \draw[dashed,fill=gray!30,thin] (0.05,.1) -- (0.05,0.95) -- (.9,.95) -- cycle;
        \draw[fill=gray!50,thin] (0.15,.3) -- (0.15,0.85) -- (.7,.85) -- cycle;
        \node[] at (0.8,0.25) {\scriptsize \textcolor{black}{$\CPre^\delta_\pi(N\cap\slice)$}};
        \draw [-latex',thin] (0.3,0.25) to [out=-170,in=-70,looseness=1.9] (0.1,0.39);
        
        \node[] at (0.5,1.3) {\scriptsize $N \cap \mathsf{slice}\subseteq \CPre^\delta_\pi(N \cap \mathsf{slice})$};
        \node[] at (0.5,-0.25) {(3)};
        
        \draw[use as bounding box, very thick,red!33!blue] (0.05,0.38) -- (0.62,0.95);
		\draw (0,0) node[circle,inner sep=2pt,draw=black,fill=blue!50] (A1) {}; 
		\draw (1,1) node[circle,inner sep=2pt,draw=black,fill=blue!50] (B1) {}; 
		\draw (0,1) node[circle,inner sep=2pt,draw=black,fill=red!50] (C1) {}; 
		\draw[dotted,use as bounding box] (A1) -- (B1) -- (C1) -- (A1);
	\end{scope}
	\end{tikzpicture} 
 \end{minipage}
  \end{center}
  \vspace{-2mm}
  \caption{\label{fig:contWins}Proof schema of Proposition~\ref{prop:contWins} applied to the corner partition of \figurename~\ref{fig:cluster}.}
  \vspace{-.5cm}
\end{figure}

Moreover, we note that a cycle that only contains punctual transitions can trivially be iterated forever by $\cont$ as there are no perturbations happening. However, such cycles always have fully-partitioned folded orbit graphs such as the rightmost example of \figurename~\ref{fig:slice}, and therefore are always robustly iterable.

\begin{restatable}[]{proposition}{pertWins}
\label{prop:pertWins}
If there is %
no winning lasso that is robustly iterable in the region automaton of $\TA$, then
$\cont$ cannot win $\thegame$ for any $\delta>0$.
\end{restatable}

\begin{figure}[t!]
  \begin{center}
    \vspace{-4mm}    
  \begin{tikzpicture}[node distance=2.5cm,auto, scale = .7]
      \tikzstyle{every state}=[thick, circle,minimum size=17pt,inner sep=0pt]
      \node (preA) at (-1,0) {};
      \node[state,] at (0,0) (A) {$\ell_0$};
      \node[state, right of=A] (B) {$\ell_1$};
      \node[accepting,state,right of=B] (C) {$\ell_2$};
      \everymath{\scriptstyle}
      \path[-latex'] 
      (preA) edge node[left] {$0<x<y<1$}(A)
      (A) edge node[below] {$0<y<1, y:=0$} (B)
      (C) edge[bend right] node[above] {$y<1$} (A)
      (B) edge node[below] {$x=1, x:=0$} (C);
    \end{tikzpicture}
    \hspace{2cm}
  \begin{tikzpicture}[line width=1pt]
    \begin{scope}
		\fill[gray!30] (0,0) -- (1,1) -- (0,1) -- cycle;
		\draw (0,0) node[circle,inner sep=1pt,fill=black] (A1) {}; 
		\draw (1,1) node[circle,inner sep=1pt,fill=black] (B1) {}; 
		\draw (0,1) node[circle,inner sep=1pt,fill=black] (C1) {}; 
		\draw[dotted,use as bounding box] (A1) -- (B1) -- (C1) -- (A1);
        \node[] at (.25, -.2)   (c0) {$c_0$};
        \node[] at (-.25, 1)   (c1) {$c_1$};  
        \node[] at (1.3, 1)   (c2) {$c_2$};  
		\begin{scope}[shorten >=1mm,shorten <=1mm,color=black,line width=.5pt]
		\draw[-latex'] (A1) .. controls +(180:6mm) and +(225:6mm) .. (A1);
		\draw[-latex'] (B1) .. controls +(45:6mm) and +(90:6mm) .. (B1);
		\draw[-latex'] (C1) .. controls +(45:6mm) and +(90:6mm) .. (C1);
		\draw[-latex'] (A1) .. controls +(22.5:4mm) and +(-112.5:4mm) .. (B1);
        \draw[-latex'] (A1) .. controls +(112.5:4mm) and +(-112.5:4mm) .. (C1);
        \draw[-latex'] (B1) .. controls +(-157.5:4mm) and +(67.5:4mm) .. (A1);
		\end{scope}
	\end{scope}
	\end{tikzpicture}
  \vspace{-5mm}    
  \end{center}
  \caption{\label{fig:pertWins}A cycle that is not robustly iterable, and its folded orbit graph.}
  \vspace{-.5cm}
\end{figure}
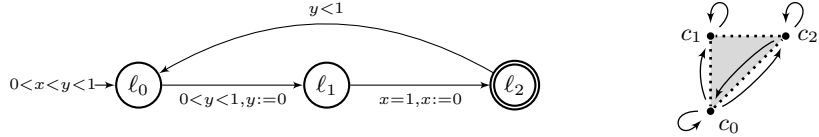
\begin{proof}[Sketch]
To establish the above proposition, we first show that if $\cont$ tries to win by repeating the same accepting region cycle then this iteration will not be possible forever.
Indeed, let $\pi$ be a region cycle and assume that $\Gamma(\pi)$ is not a cluster graph,
then there must exist at least two weakly connected SCCs $I$ and $J$ such that $J$ is reachable from $I$ but not vice versa, cf. \figurename~\ref{fig:pertWins} where $I=\{c_0,c_2\}$ and $J=\{c_1\}$.
Let $\nu$ be a valuation in the region $r$ from \figurename~\ref{fig:pertWins}. %
We show that $\cont$ cannot iterate $\pi$ forever when starting from $\nu$.

Let $\vecl=(\lambda_0, \lambda_1, \lambda_2)$ be the corner weights of $\nu$. We first use a result from~\cite{SBMR-concur13} that exhibits a precise strategy for $\pert$ which inflicts a well chosen perturbation $\varepsilon$ that pushes any run away from $\{c_0,c_2\}$ and towards $\{c_1\}$. 
This entails the following behavior: after each visit of $r$, one reaches a valuation $\nu'$ of corner weights $\vecl' = (\lambda_0', \lambda_1', \lambda_2')$ such that
$\lambda_0 + \lambda_2 > \lambda'_0 + \lambda'_2 + \varepsilon$. A direct corollary of this observation is that under this strategy of $\pert$, any run from $\nu$ that tries to visit $r$ repeatedly will ultimately drift out of the region $r$.
This is detailed in Appendix~\ref{app:pertWins}. 
In order to conclude, we still need to argue that $\cont$ cannot win by switching between different region cycles. But since there exists only a finite number of folded orbit graphs, one can use \emph{Ramsey-like} arguments to factorize any infinite run into a finite prefix and factors that all share the same folded orbit graph. Assuming that no such folded orbit graph is a cluster graph, together with the previous intuition, entails Proposition~\ref{prop:pertWins}.
\end{proof}

\section{Conclusion}

This work is a first technical step towards a more general setting where timing measurements in components of the system could either be reliable or disruptable.
As future work, we plan to investigate the following directions.
Extend the current work by allowing some non-punctual transitions to be reliable (\textit{i.e.}~not under perturbation in the game semantics), 
and define criteria for repairing timed automata that are not robustly controllable in the classical sense
by transforming them into timed automata with exact components that are robust in the setting addressed in this paper.

\bibliographystyle{splncs04}
\bibliography{ref}

\begin{thebibliography}{10}
\providecommand{\url}[1]{\texttt{#1}}
\providecommand{\urlprefix}{URL }
\providecommand{\doi}[1]{https://doi.org/#1}

\bibitem{AD-tcs94}
Alur, R., Dill, D.L.: A theory of timed automata. Theoretical Computer Science
  \textbf{126}(2),  183--235 (1994)

\bibitem{AsarinMPS98}
Asarin, E., Maler, O., Pnueli, A., Sifakis, J.: Controller synthesis for timed
  automata. In: Proceedings of IFAC Symposium on System Structure and Control.
  pp. 469--474. Elsevier (1998)

\bibitem{AB-formats11}
Basset, N., Asarin, E.: Thin and thick timed regular languages. In: Fahrenberg,
  U., Tripakis, S. (eds.) Formal Modeling and Analysis of Timed Systems.
  Lecture Notes in Computer Science, vol.~6919, pp. 113--128. Springer (2011)

\bibitem{bengtsson2003timed}
Bengtsson, J., Yi, W.: Timed automata: Semantics, algorithms and tools. In:
  Advanced Course on Petri Nets. pp. 87--124. Springer (2003)

\bibitem{BMRS19}
Busatto{-}Gaston, D., Monmege, B., Reynier, P., Sankur, O.: Robust controller
  synthesis in timed b{\"{u}}chi automata: {A} symbolic approach. In: Dillig,
  I., Tasiran, S. (eds.) Computer Aided Verification - 31st International
  Conference, {CAV} 2019, New York City, NY, USA, July 15-18, 2019,
  Proceedings, Part {I}. Lecture Notes in Computer Science, vol. 11561, pp.
  572--590. Springer (2019). \doi{10.1007/978-3-030-25540-4\_33}

\bibitem{CassezDFLL05}
Cassez, F., David, A., Fleury, E., Larsen, K.G., Lime, D.: Efficient on-the-fly
  algorithms for the analysis of timed games. In: Abadi, M., de~Alfaro, L.
  (eds.) {CONCUR} 2005 - Concurrency Theory, 16th International Conference,
  {CONCUR} 2005, San Francisco, CA, USA, August 23-26, 2005, Proceedings.
  Lecture Notes in Computer Science, vol.~3653, pp. 66--80. Springer (2005).
  \doi{10.1007/11539452\_9}

\bibitem{GuptaHenzingerJagadeesan97}
Gupta, V., Henzinger, T.A., Jagadeesan, R.: Robust timed automata. In: Maler,
  O. (ed.) Hybrid and Real-Time Systems, International Workshop. HART'97,
  Grenoble, France, March 26-28, 1997, Proceedings. Lecture Notes in Computer
  Science, vol.~1201, pp. 331--345. Springer (1997). \doi{10.1007/BFB0014736}

\bibitem{MalerPS95}
Maler, O., Pnueli, A., Sifakis, J.: On the synthesis of discrete controllers
  for timed systems (an extended abstract). In: Mayr, E.W., Puech, C. (eds.)
  {STACS} 95, 12th Annual Symposium on Theoretical Aspects of Computer Science,
  Munich, Germany, March 2-4, 1995, Proceedings. Lecture Notes in Computer
  Science, vol.~900, pp. 229--242. Springer (1995).
  \doi{10.1007/3-540-59042-0\_76}

\bibitem{OualhadjRS14}
Oualhadj, Y., Reynier, P., Sankur, O.: Probabilistic robust timed games. In:
  Baldan, P., Gorla, D. (eds.) {CONCUR} 2014 - Concurrency Theory - 25th
  International Conference, {CONCUR} 2014, Rome, Italy, September 2-5, 2014.
  Proceedings. Lecture Notes in Computer Science, vol.~8704, pp. 203--217.
  Springer (2014). \doi{10.1007/978-3-662-44584-6\_15}

\bibitem{Puri00}
Puri, A.: Dynamical properties of timed automata. Discret. Event Dyn. Syst.
  \textbf{10}(1-2),  87--113 (2000). \doi{10.1023/A:1008387132377}

\bibitem{Rodriguez-NavasP13}
Rodr{\'{\i}}guez{-}Navas, G., Proenza, J.: Using timed automata for modeling
  distributed systems with clocks: Challenges and solutions. {IEEE} Trans.
  Software Eng.  \textbf{39}(6),  857--868 (2013). \doi{10.1109/TSE.2012.73}

\bibitem{Sankur13}
Sankur, O.: Robustness in timed automata : analysis, synthesis, implementation.
  (Robustesse dans les automates temporis{\'{e}}s : analyse, synth{\`{e}}se,
  impl{\'{e}}mentation). Ph.D. thesis, {\'{E}}cole normale sup{\'{e}}rieure de
  Cachan, Paris, France (2013),
  \url{https://tel.archives-ouvertes.fr/tel-00910333}

\bibitem{SankurBM14}
Sankur, O., Bouyer, P., Markey, N.: Shrinking timed automata. Inf. Comput.
  \textbf{234},  107--132 (2014). \doi{10.1016/J.IC.2014.01.002}

\bibitem{SBMR-concur13}
Sankur, O., Bouyer, P., Markey, N., Reynier, P.A.: Robust controller synthesis
  in timed automata. In: D'Argenio, P.R., Melgratti, H. (eds.) {P}roceedings of
  the 24th {I}nternational {C}onference on {C}oncurrency {T}heory
  ({CONCUR}'13). Lecture Notes in Computer Science, vol.~8052, pp. 546--560.
  Springer (2013). \doi{10.1007/978-3-642-40184-8_38}

\bibitem{DBLP:conf/formats/Stainer12}
Stainer, A.: Frequencies in forgetful timed automata. In: Jurdzinski, M.,
  Nickovic, D. (eds.) Formal Modeling and Analysis of Timed Systems - 10th
  International Conference, {FORMATS} 2012, London, UK, September 18-20, 2012.
  Proceedings. Lecture Notes in Computer Science, vol.~7595, pp. 236--251.
  Springer (2012). \doi{10.1007/978-3-642-33365-1\_17}

\bibitem{WulfDMR04}
Wulf, M.D., Doyen, L., Markey, N., Raskin, J.: Robustness and implementability
  of timed automata. In: Lakhnech, Y., Yovine, S. (eds.) Formal Techniques,
  Modelling and Analysis of Timed and Fault-Tolerant Systems, Joint
  International Conferences on Formal Modelling and Analysis of Timed Systems,
  {FORMATS} 2004 and Formal Techniques in Real-Time and Fault-Tolerant Systems,
  {FTRTFT} 2004, Grenoble, France, September 22-24, 2004, Proceedings. Lecture
  Notes in Computer Science, vol.~3253, pp. 118--133. Springer (2004).
  \doi{10.1007/978-3-540-30206-3\_10}

\end{thebibliography}
\newpage
\appendix

\section{Classical results on regions and orbit graphs.}\label{app:fog}

\begin{lemma}[\cite{AD-tcs94}]\label{lm:bissimulation}
    Let $(\ell_1,\nu_1)$ and $(\ell_2,\nu_2)$ be configurations of a timed automaton $\TA$ and $r_1,r_2$ be %
    regions so that $\nu_1\in r_1$ and $\nu_2\in r_2$. 
    The following are equivalent:
    \begin{itemize}
        \item there is a run from $(\ell_1,\nu_1)$ to $(\ell_2,\nu_2)$ in $\TA$,
        \item for all $\nu_1'\in r_1$ there exists $\nu_2'\in r_2$ and a run from $(\ell_1,\nu_1')$ to $(\ell_2,\nu_2')$, and
        \item for all $\nu_2'\in r_2$ there exists $\nu_1'\in r_1$ and a run from $(\ell_1,\nu_1')$ to $(\ell_2,\nu_2')$.
    \end{itemize}
\end{lemma}

As such, we will abstract sets of runs that go through the same regions and the same edges of $\TA$ as paths from region to region in a finite automaton called the region abstraction of $\TA$.

\begin{lemma}[\cite{Puri00}]\label{lm:orbit-arity}
    Let $\pi$ be a region cycle around a region $r$ of corners $\mathcal C$.
    For every corner $c\in\mathcal C$, there is a corner $c'\in\mathcal C$ so that there is an edge from $c$ to $c'$ in $\Gamma(\pi)$.
    Similarly, for every corner $c'\in\mathcal C$, there is a corner $c\in\mathcal C$ so that there is an edge from $c$ to $c'$ in $\Gamma(\pi)$.
\end{lemma}

\begin{lemma}\label{lm:orbit-graph-distribution}
    Let $\pi$ be a well-formed region path from a region state $(\ell,r)$ of dimension $m$ to a region state $(\ell',r')$ of dimension $m'$, of respective corners $\mathcal C=\{c_0,\dots c_m\}$ and $\mathcal C'=\{c_0',\dots c_{m'}'\}$. Let $\nu\in r$ and $\nu'\in r'$ be valuations of respective corner weights $\vecl=(\lambda_0,\dots,\lambda_m)$ and $\vecl'=(\lambda_0',\dots,\lambda_{m'}')$.
    For each $0\leq i\leq m$, let $\mathsf{post}_i\subseteq\mathcal C'$ be the corners $c_j'$ so that there is a path in $\gamma(\pi)$ from $(\mathsf{start},c_i)$ to $(\mathsf{end},c_j')$.
    Conversely, for each $0\leq j\leq m'$ let $\mathsf{pre}_j\subseteq\mathcal C$ be the corners $c_i$ so that there is a path in $\gamma(\pi)$ from $(\mathsf{start},c_i)$ to $(\mathsf{end},c_j')$.

    Then, there is a run from $(\ell,\nu)$ to $(\ell',\nu')$ that follows $\pi$ if and only if there is for each $c_i\in\mathcal C$ a probability distribution $\vec{p}_i:\mathsf{post}_i\to [0,1]$ so that %
    for each $0\leq j\leq m'$, $\lambda_j' = \sum_{c_i\in\mathsf{pre}_j} \lambda_i \vec{p}_i(c_j')$.
\end{lemma}

The above lemma entails the following characterization of the reachability relation between the valuations of the same region.

\begin{corollary}[\cite{Puri00}]\label{cor:fog-distribution}
    Let $\pi$ be a well-formed region cycle around a region state $(\ell,r)$ of dimension $m$, of corners $\mathcal C=\{c_0,\dots c_m\}$. Let $\nu\in r$ and $\nu'\in r'$ be valuations of respective corner weights $\vecl=(\lambda_0,\dots,\lambda_m)$ and $\vecl'=(\lambda_0',\dots,\lambda_{m'}')$.
    For each $0\leq i\leq m$, let $\mathsf{post}_i\subseteq\mathcal C$ be the corners $c_j$ so that there is an edge in $\Gamma(\pi)$ from $c_i$ to $c_j$.
    Conversely, for each $0\leq j\leq m$ let $\mathsf{pre}_j\subseteq\mathcal C$ be the corners $c_i$ so that there is an edge in $\Gamma(\pi)$ from $c_i$ to $c_j$.

    Then, there is a run from $\nu$ to $\nu'$ that follows $\pi$ if and only if there is for each $c_i\in\mathcal C$ a probability distribution $\vec{p}_i:\mathsf{post}_i\to [0,1]$ so that 
    for each $0\leq j\leq m$, $\lambda_j' = \sum_{c_i\in\mathsf{pre}_j} \lambda_i \vec{p}_i(c_j)$.
\end{corollary}

\newpage
\section{Proof of Lemma~\ref{lm:robCycle}}
\robCycle*

This property also holds if we replace $\pi^k$ by the concatenation $\pi_1\pi_2\dots\pi_k$, where $\pi_1,\dots,\pi_k$ are finite robust paths so that $\Gamma(\pi)=\Gamma(\pi_1)=\dots=\Gamma(\pi_k)$.

\begin{proof}
    We start by showing that the two cases are mutually exclusive. 
    If $\Gamma(\pi^k)$ is a cluster graph, then for every $l\geq 1$, the cycle $\pi^{kl}$ has the same folded orbit graph as $\pi^k$ (by decomposition of $\gamma(\pi^{kl})$ into concatenated orbit graphs $\gamma(\pi^k)\dots\gamma(\pi^k)$), and therefore is a cluster graph as well.
    Similarly, if $\Gamma(\pi^{k'})$ contains a weakly connected component that is not strongly connected, then for every $l'\geq 1$,  $\Gamma(\pi^{k'l'})$ also has a weakly connected component that is not strongly connected.
    Indeed, if there is an SCC $I$ in $\Gamma(\pi^{k'})$, a corner $c\in I$ and a corner $c'\not\in I$, so that there is an edge from $c$ to $c'$ in $\Gamma(\pi^{k'})$ but no path from $c'$ to $c$, then by decomposition of $\gamma(\pi^{k'l'})$ into concatenated orbit graphs $\gamma(\pi^{k'})\dots\gamma(\pi^{k'})$, there is some corner $c''$ reached from $c'$ in $\gamma(\pi^{k'})$ by a path of length $l'-1$, so that there is an edge from $c$ to $c''$ in $\Gamma(\pi^{k'l'})$ but no path from $c''$ to $c$.
    Thus, if by contradiction both conditions were true for $k$ and $k'$, respectively, then $\Gamma(\pi^{kk'})$ would be a cluster graph with a weakly connected component that is not strongly connected, which is impossible.
    
    Now, assume that for every $1\leq k\leq (m+1)!$, every weakly connected component of $\Gamma(\pi^k)$ is also strongly connected, \textit{i.e.}~$\Gamma(\pi^k)$ is a disjoint union of SCCs.
    Let us show that $\Gamma(\pi^{m(m+1)!})$ must be a cluster graph. 
    First, we show that there is a self-loop on every corner in $\Gamma(\pi^{(m+1)!})$.
    Indeed, $\Gamma(\pi)$ is a disjoint union of SCCs, and by Lemma~\ref{lm:orbit-arity} every corner has an outgoing edge in $\Gamma(\pi)$, so that every corner must belong to some cycle of $\Gamma(\pi)$ of length $1\leq l\leq m+1$.
    Then, there is a path from $(\mathsf{start},c)$ to $(\mathsf{end},c)$ in $\gamma(\pi^{l})$, and thus in $\gamma(\pi^{(m+1)!})$ as $(m+1)!$ is a multiple of $l$.
    Thus, for all $c$ there is a self-loop from $c$ to $c$ in $\Gamma(\pi^{(m+1)!})$,
    and by the starting assumption $\Gamma(\pi^{(m+1)!})$ is also a disjoint union of SCCs.
    Let $c$ and $c'$ be corners in the same SCC of $\Gamma(\pi^{(m+1)!})$.
    Then, there exists in $\Gamma(\pi^{(m+1)!})$ a path from $c$ to $c'$ of length $l\leq m$, and a path of length $m-l$ from $c'$ to $c'$ (by repeating the self loop on $c'$), so that there is a path of length $m$ from $c$ to $c'$ in $\Gamma(\pi^{(m+1)!})$.
    It follows that there is an edge from $c$ to $c'$ in $\Gamma(\pi^{m(m+1)!})$ for all $c,c'$ in the same SCC of $\Gamma(\pi^{(m+1)!})$.
    Moreover, for any $c,c'$ in disjoint (and thus independent) SCCs of $\Gamma(\pi^{(m+1)!})$ ,there can be no edge from $c$ to $c'$ in $\Gamma(\pi^{m(m+1)!})$, so that overall $\Gamma(\pi^{m(m+1)!})$ is a cluster graph where every SCC of $\Gamma(\pi^{(m+1)!})$ is independent and complete.
    
    Finally, we note that replacing $\pi^k$ by the concatenation $\pi_1\pi_2\dots\pi_k$ of paths with the same folded orbit graph does not affect the previous proof, that entirely relies on the edges of $\Gamma(\pi)$ and $\Gamma(\pi^k)$. Indeed, we note by definition of folded orbit graphs that if $\Gamma(\pi)=\Gamma(\pi_1)=\dots=\Gamma(\pi_k)$ then $\Gamma(\pi_1\dots\pi_k)=\Gamma(\pi^k)$.
\end{proof}

\section{Proofs of Section~\ref{sec:robust-path}} %

Let $M,N$ be DBMs.
We say that $N\sqsubseteq M$ if $N\subseteq M$ and the constants in $M$ and $N$ are equal.
In other words, $N$ is obtained from $M$ by making some large inequalities strict, but has the same interior.
If in particular $M$ is a DBM that encodes a region $r$, it holds that $N\sqsubseteq M$ and $N\neq\emptyset$ implies $N=r$, as every constraint that encodes a region is either strict or an equality constraint that cannot be made strict without becoming empty.

The next two results are proven in \cite[Lemma~8.6.1]{Sankur13} and \cite[Lemma~8.6.2]{Sankur13} in a setting without punctual edges or regions. %
We prove them by induction on the length of $\pi$, by noting that they are stable by composition, so that only well-formed atomic paths need to be considered.
In particular, $\CPre_\pi^\delta$ with $\pi$ an atomic robust path with a non-punctual edge can be decomposed into the DBM operations of Definition~\ref{def:dbm-operations},
and each of these operations satisfies the two results individually by~\cite{Sankur13}.
If $\pi$ is an atomic robust path with a punctual edge, then by definition of $\CPre$ in this case we can also decompose $\CPre_\pi^\delta$ into the same elementary DBM operations, and conclude similarly.
\begin{lemma}\label{lm:cpredeltaslice}
    Let $\pi$ be a robust region path from $r$ to $r'$, and let $\delta>0$.
    Let $M$ and $M'$ be DBMs
    such that $M\neq\emptyset$ and $M'\neq\emptyset$.
    Moreover, assume that $M=\CPre_\pi^0(M')$.
    Then, for any DBM $N'\sqsubseteq M'$ and
    for any shrinking matrix $P'$, there exists a DBM $N\sqsubseteq M$ and a shrinking matrix $P$ such that $N-\delta P=\CPre_\pi^\delta((N'-\delta P'))$.
\end{lemma}

\begin{lemma}\label{lm:ballslice}
    Let $\pi$ be a robust region path from $r$ to $r'$.
    Let $N$ be a DBM %
    so that
    there exists $\nu\in r'$ and $\varepsilon>0$ with
    $(\Ball_{d_\infty}(\nu,\varepsilon)\cap r')\subseteq N$. 
    Then, $\CPre_\pi^{\delta}(N)$ is non-empty
    for a small enough $\delta>0$, and in fact contains $\Ball_{d_\infty}(\nu',\varepsilon')\cap r$ for some $\nu'\in r$ and $\varepsilon'>0$.
\end{lemma}

\robpath*

\begin{proof}
    If $\pi$ is not robust, then it contains some sub-path $\pi'=(\ell,r)\xrightarrow{\mathsf{delay}}(\ell,r')\xrightarrow{g,R}(\ell',r'')$ where $r'$ is punctual and $g$ is non-punctual.
    Then, $\CPre_{\pi'}^\delta(s)=\{\nu \in r \mid \exists d\geq \delta, \forall d'\in[d-\delta,d+\delta], \nu+d'\in r' \land (\nu+d')[R:=0]\in s\}=\emptyset$ as by definition of punctuality $\nu\in r'$ implies $\nu+d'\not\in r'$ for $d'\neq 0$.
    Then, $\CPre_{\pi}^\delta(s)=\emptyset$ by composition.
    
    If $\pi$ is a robust region path from $(\ell,r)$ to $(\ell',r')$, we apply Lemma~\ref{lm:ballslice} on $N=r'$ to justify $\CPre_\pi^\delta(r')\neq\emptyset$ for $\delta$ small enough.
    In particular, the existence of $\nu,\varepsilon$ such that $\Ball_{d_\infty}(\nu,\varepsilon)\cap r'\subseteq N$ is trivial
    as $(A\cap r')\subseteq r'$ for any set $A$.
\end{proof}

\newpage
\section{Properties of slices}
\label{app:slices}

\begin{lemma}\label{lm:slicepartition}
Let $C_0\uplus \dots\uplus C_k$ be a corner partition of a region $r$ into $k+1$ colors.
The set of slices $\slice^r_{C_0\uplus \dots\uplus C_k}(\vecw)$, where $\vecw$ ranges over every color weight vector,
partitions $r$ into  non-empty slices.
\end{lemma}
\begin{proof}
For a given corner partition $C_0\uplus \dots\uplus C_k$, for any valuation $\nu$ in $r$
of corner weights $\vecl=(\lambda_0,\dots,\lambda_m)$
there is a unique weight vector $\vecw=(w_1,\dots, w_{k})$ such that $\nu\in\slice^r_{C_0\uplus \dots\uplus C_k}(\vecw)$: the one defined by $w_j=\sum\limits_{\substack{0\leq i\leq m\\\text{s.t. } c_i\in C_j}} \lambda_i$ for all $0\leq j\leq k$.
Conversely, for every color weight vector $\vecw=(w_1,\dots, w_{k})$, $\slice^r_{C_0\uplus \dots\uplus C_k}(\vecw)$ is not empty: it contains at least the valuation $\nu$ of corner weights $\vecl=(\lambda_0,\dots,\lambda_m)$, so that for each $0\leq i\leq m$, if $c_i\in C_j$ in the corner partition then $\lambda_i=\frac{w_j}{|C_j|}$.
\end{proof}

We show that slices defined from a cluster corner partition are zones.
\begin{restatable}[]{proposition}{slicezones}
    \label{prop:slicezones}
    Let $\pi$ be a robustly iterable region cycle around a region $r$ of dimension $m$.
    Then, for any rational color weight vector $\vecw$, the slice $\slice_{\pi}(\vecw)$ is a DBM.
    More precisely, $\slice_\pi(\vecw)$ is equal to the intersection of $r$ and a set of atomic constraints of the shape $x=c$ or $x-y=c$, with $x,y\in\Clocks$ and rational constants $c$.    
\end{restatable}

In order to prove Proposition~\ref{prop:slicezones}, we need to describe cluster corner partitions in more details:
\begin{lemma}
    \label{lm:partition}
    Let $\pi$ be a robustly iterable region cycle around a region $r$ of dimension $m$.
    Let $\mathcal{C}=\{c_0,\dots,c_m\}$ be the corners of $r$, so that $\|c_0\|_1<\|c_1\|_1<\dots<\|c_m\|_1$.
    Let $C_0\uplus \dots\uplus C_k$ be the cluster corner partition of $r$ induced by $\Gamma(\pi)$.
    There exists a set of splitting positions $P=\{p_0,\dots,p_k\}\subseteq\{0,\dots,m\}$ with $p_0<\dots<p_k$, so that
    $C_0\uplus \dots\uplus C_k$ can be ordered such that
    \begin{itemize}
        \item $C_0=\{c_i \mid i\in [0,p_0]\cup(p_k,m]\}$, and
        \item for all $1\leq j \leq k$, $C_j=\{c_i \mid i\in (p_{j-1},p_{j}]\}$
    \end{itemize}
\end{lemma}

\begin{proof}
    We show this splitting position-based characterisation of the cluster corner partition by generalizing it to region paths that may not be cycles, then by showing this generalisation by induction on the length of the path, and finally by applying it on a cycle.
    
    Let $\pi$ be a robust region path from a region $r$ of dimension $m$ to a region $r'$ of dimension $m'$.
    Let $\mathcal{C}=\{c_0,\dots,c_{m}\}$ (resp. $\mathcal{C}'=\{c_0',\dots,c_{m'}'\}$) be the corners of $r$ (resp. $r'$),
    so that $\|c_0\|_1<\|c_1\|_1<\dots<\|c_m\|_1$ and $\|c_0'\|_1<\|c_1'\|_1<\dots<\|c_{m'}'\|_1$.
    The weakly connected components of $\gamma(\pi)$ naturally induce a partition of its vertices into $k+1$ colors,
    so that every vertex $(\mathsf{step}_i,c)$ has color $0\leq j\leq k$ if it belongs 
    to the $j$-th weakly connected component.
    This partition of $\gamma(\pi)$ naturally induces two corner partitions $C_0\uplus \dots\uplus C_k$ and $C_0'\uplus \dots\uplus C_k'$ of $r$ and $r'$, respectively, based on the colors of $\mathsf{start}$ and $\mathsf{end}$ vertices.

\newpage
    We show by induction on the length of $\pi$ that
    there exists a set of splitting positions $P=\{p_0,\dots,p_k\}\subseteq\{0,\dots,m\}$ with $p_0<\dots<p_k$, so that
    $C_0\uplus \dots\uplus C_k$ can be ordered such that
    \begin{itemize}
        \item $C_0=\{c_i \mid i\in [0,p_0]\cup(p_k,m]\}$, and
        \item for all $1\leq j \leq k$, $C_j=\{c_i \mid i\in (p_{j-1},p_{j}]\}$
    \end{itemize}
    and there exists a set of splitting positions $P'=\{p_0',\dots,p_k'\}\subseteq\{0,\dots,m'\}$ with $p_0'<\dots<p_k'$, so that
    $C_0'\uplus \dots\uplus C_k'$ can be ordered such that
    \begin{itemize}
        \item $C_0'=\{c_i' \mid i\in [0,p_0']\cup(p_k',m']\}$, and
        \item for all $1\leq j \leq k$, $C_j'=\{c_i' \mid i\in (p_{j-1}',p_{j}']\}$
    \end{itemize}
    
    We start by considering the case where $\pi$ is a single delay transition from $r$ to some time-successor $r'$.
    We note the following properties of $\gamma(\pi)$:
    \begin{itemize}
        \item every corner in $r$ has a unique time-successor in $r'$, except for one corner that can have $c_0'$ and $c_{m'}'$ as successors if $r'$ is non-punctual, and
        \item every corner in $r'$ has a unique time-predecessor in $r$, except for one that can have $c_0$ and $c_{m}$ as predecessors if $r$ is non-punctual.
    \end{itemize}
    Then, the weakly connected components of $\gamma(\pi)$ induce partitions that are either singletons, $\{c_0,c_m\}$ or $\{c_0',c_{m'}'\}$. These partitions can thus be described with consecutive splitting positions for the singletons, $p_0=0$ and $p_k=m-1$ for $\{c_0,c_m\}$ if $r$ is non-punctual, $p_0=0$ and $p_k=m$ if $r$ is punctual, and similarly for $r'$.
    
    We now consider the case where $\pi$ is a single edge transition from $r$ to $r'$ of reset set $R$.
    We note the following properties of $\gamma(\pi)$:
    \begin{itemize}
        \item every corner $c$ in $r$ has a unique successor $c[R:=0]$ in $r'$,
        \item if a corner $c'$ in $r'$ has two predecessors $c_i$ and $c_j$ in $r$ with $\|c_i\|_1<\|c_j\|_1$, then every corner $c$ so that $\|c_i\|_1<\|c\|_1<\|c_j\|_1$ is also a predecessor of $c'$.
    \end{itemize}
    This last property holds by definition of corners as $c_i[R:=0]=c_j[R:=0]$ implies that $R$ contains every clock that differentiates them in the clock order of the region, so that every corner $c$ in between is also reset to $c'$ as it only differs from $c_i$ and $c_j$ on the same clocks of $R$.
    Then, the weakly connected components of $\gamma(\pi)$ induce partitions that are all singletons in $r'$, and either singletons or intervals of consecutive corners in $r$, These partitions can thus be described with splitting positions, with $p_k=m$ and $p_k'=m'$.
    
    Assume now that the inductive property holds on two paths that can be concatenated. 
    Whenever a region path $\pi$ is concatenated to another path $\pi'$, the operation can merge weakly connected components of their orbit graphs, thus reducing the number of colors.
    This happens at the junction of $\gamma(\pi)$ and $\gamma(\pi')$, where colors of each side that intersect merge to become a single color.
    These new colors are obtained as unions of intervals of corners of index $(p_{j-1},p_{j}]$ and $(p_{j'-1}',p_{j'}']$ that intersect, which can always be expressed as new intervals of corners of the shape $(p_{j''-1}'',p_{j''}'']$.
    
    The case of the corners of index in $[0,p_0]\cup(p_k,m]$ leads to the same conclusion, as this interval shape is also stable by union of intersecting intervals.
    This concludes the inductive proof.
    
    Finally, assume $\pi$ is a robustly iterable cycle.
    Then, In order to obtain the cluster corner partition induced by $\Gamma(\pi)$ from the partitions $C_0\uplus \dots\uplus C_k$ and $C_0'\uplus \dots\uplus C_k'$ induced by $\gamma(\pi)$ as described above, one needs to do one last merge of colors, with colors that intersect if we merge $\mathsf{start}$ and $\mathsf{end}$ vertices of $\gamma(\pi)$ merging to become colors of $\Gamma(\pi)$.
    This operation obeys the same principle as the concatenation of paths described above, and maintains the shape of the color partition as characterised by a set of splitting positions.
\end{proof}

\begin{proof}[Proof of Proposition~\ref{prop:slicezones}]
Let $\pi$ be a robustly iterable region cycle around a region $r$.
Let $r$ be of dimension $m$, with $\mathcal{C}=\{c_0,\dots,c_m\}$ the set of its corners, so that $\|c_0\|_1<\|c_1\|_1<\dots<\|c_m\|_1$.
Let $C_0\uplus \dots\uplus C_k$ be the cluster corner partition induced by $\Gamma(\pi)$.
Assume w.l.o.g. that the smallest corner $c_0$ is in $C_0$.
For any given weight vector $\vecw$, let us detail a system of equations that encodes $\slice_\pi(\vecw)$.

Let $r$ be a region.
If $r$ is a singleton that contains a single valuation $\nu$ ($r$ is a region of dimension $0$), then $r$ is a closed set, $r$ has only one corner,
the only corner partition of $r$ has one set ($k=1$), and the only slice is equal to $r$.
In the following, we assume that $r$ has dimension at least $1$. %

The region $r$ of dimension $m>0$ is characterized by $(\iota, \beta)$ where $\beta$ is a partition of $\mathcal X$ into $\beta_0\uplus \beta_1\uplus \cdots \uplus \beta_m$ with $\forall 1\leq j\leq m,\beta_j\neq\emptyset$.
Recall that the corners of $r$ are the $m+1$ valuations $c_0\dots c_m$ so that for each $0\leq i\leq m$, $c_i$ is the corner such that $\forall 0\leq j\leq m-i, \forall x\in \beta_j, c_i(x)=\iota(x)$, and $\forall m-i< j\leq m, \forall x\in \beta_j, c_i(x)=\iota(x)+1$.
In particular, $c_0$ is the smallest corner of $r$ and $c_m$ the largest, with $\|c_0\|_1<\|c_1\|_1<\dots<\|c_m\|_1$.

For each $1\leq j\leq m$, fix $x_j$ a clock in $\beta_j$, picked arbitrarily.

Let $\nu$ be a valuation in $r$, with $\lambda_0,\dots,\lambda_m\in(0,1]$ its corner weights, so that $\sum_{i=0}^{m} \lambda_i=1$ and such that $\nu=\sum_{i=0}^m \lambda_i c_i$.
Then, note that for each $0\leq j\leq m$, it holds that
$c_i(x_j)=\iota(x_j)$ for all $0\leq i\leq m$ with $i\leq m-j$ and $c_i(x_j)=\iota(x_j)+1$ for all $0\leq i\leq m$ with $i> m-j$.
Then,  for each $1\leq j\leq m$, we have:
\begin{align*}
\nu(x_j)-\iota(x_j)&=\left(\sum_{i=0}^m \lambda_i c_i(x_j)\right)-\iota(x_j)\\
&=\iota(x_j)\left(\sum_{i=0}^{m-j} \lambda_i\right) + (\iota(x_j)+1)\left(\sum_{i=m-j+1}^{m} \lambda_i\right) - \iota(x_j)\left(\sum_{i=0}^{m} \lambda_i\right)\\
&=(\iota(x_j)+1)\left(\sum_{i=m-j+1}^{m} \lambda_i\right) - \iota(x_j)\left(\sum_{i=m-j+1}^{m} \lambda_i\right)\\
&=\sum_{i=m-j+1}^{m} \lambda_i
\end{align*}

For each $1\leq j\leq m$, let $f_j=\nu(x_j)-\iota(x_j)$ be the fractional part of clock $x_j$.
Moreover, we extend the notation to $j=0$ and $j=m+1$ by setting $f_0=0$ and $f_{m+1}=1$.
Then, it follows that
\begin{itemize}
\item first, $\lambda_{0}=1-\left(\sum_{i=1}^{m} \lambda_i\right)=f_{m+1}-f_m$,
\item then for each $0< i<m$, 
$\lambda_{i}=\left(\sum_{l=i}^{m} \lambda_l\right)-\left(\sum_{l=i+1}^{m} \lambda_l\right)=f_{m-i+1} - f_{m-i}$,
\item and finally, $\lambda_{m}=\left(\sum_{i=m}^{m} \lambda_i\right)-0=f_1-f_0$.
\end{itemize}
So that for all $0\leq i\leq m$, it holds that $\lambda_{i}=f_{m-i+1} - f_{m-i}$.
Thus, for all $0\leq a\leq b\leq m$, $\sum_{i=a}^b \lambda_i = f_{m-a+1} - f_{m-b}$.

Now, let $C_0\uplus \dots\uplus C_k$ be the cluster corner partition of $r$ into $k+1$ colors induced by $\Gamma(\pi)$, characterized by the splitting positions $p_0,\dots,p_k$ by Lemma~\ref{lm:partition}.
If $k=0$, there is only one color ($C_0=\mathcal C$), and the only slice of $r$ is equal to $r$.
We assume $k>0$ in the following.

Let $\vecw=(w_0,\dots, w_k)\in(0,1]^{k+1}$ be a color weight vector, so that $\|\vecw\|_1=1$, that defines a slice $\slice_\pi(\vecw)$.
Let $\nu$ be a valuation in $r$ of corner weights $\vecl=(\lambda_0,\dots,\lambda_m)$.
Then, $\nu$ belongs to the slice $\slice_\pi(\vecw)$, if 
\begin{itemize}
    \item $w_0=\sum_{i\in [0,p_0]\cup(p_k,m]} \lambda_i = f_{m+1} - f_{m-p_0} + f_{m-p_k} - f_0$, and
    \item for all $1\leq j\leq k$, $w_j=\sum_{i\in (p_{j-1},p_{j}]} \lambda_i = f_{m-p_{j-1}} - f_{m-p_j}$.
\end{itemize}

Then, for any $0\leq j\leq k$, let $\mathcal E_j$ be a clock constraint over $\Clocks$ encoding these conditions, so that 
\begin{itemize}
    \item if $p_k<m$: $\mathcal E_0$ is $x_{m-p_0} - x_{m-p_k} = 1 + \iota(x_{m-p_0}) - \iota(x_{m-p_k}) - w_0$,
    \item If $p_k=m$: $\mathcal E_0$ is $x_{m-p_0} = 1 + \iota(x_{m-p_0}) - w_0$, and
    \item for all $1\leq j\leq k$, $\mathcal E_j$ is $x_{m-p_j} - x_{m-p_{j-1}} = \iota(x_{m-p_j}) - \iota(x_{m-p_{j-1}}) -w_j $.
\end{itemize}

Note that these constraint $\mathcal E_j$ are all equality constraints of the shape $x-y=c-w_i$,
with $c$ a constant in $\Nat$, except for $\mathcal E_0$ in the case where $p_k=m$ (\textit{i.e.} when the corner partition is non-diagonal), which is a constraint of the shape $x=c-w_0$.
This concludes the proof of Proposition~\ref{prop:slicezones}.
\end{proof}

\section{Proof of Proposition~\ref{prop:contWins}}\label{app:slice-based-proof}

\begin{lemma}[{\cite[Lemma~8.6.4]{Sankur13}}]\label{lm:initslice}
    Let $r$ be a region.
    There exists a DBM $N$
    such that there exists $\nu\in r$ and $\varepsilon>0$ with
    $(\Ball_{d_\infty}(\nu,\varepsilon)\cap r) \subseteq N$, 
    and such that for any shrinking matrix $Q$ so that there exists $\delta>0$ with $r-\delta Q\neq \emptyset$, we have $\exists \delta'>0$, $N\subseteq (r-\delta' Q)$.
\end{lemma}

Valuations that can reach each other must belong to the same slice:
\begin{lemma}\label{lm:slicedynamics}
Let $\pi$ be a robustly iterable region cycle. 
Let $\vecw$ and $\vecw'$ be weight vectors, and $\nu,\nu'$ be valuations of $r$
so that $\nu\in\slice_\pi(\vecw)$ and $\nu'\in\slice_\pi(\vecw')$.
If there is a run from $\nu$ to $\nu'$ that follows $\pi$, then $\vecw=\vecw'$.
\end{lemma}
\begin{proof}
By Corollary~\ref{cor:fog-distribution}, since there is a run from $\nu$ to $\nu'$ there is for each $c_i\in\mathcal C$ a probability distribution $\vec{p}_i:\mathsf{post}_i\to [0,1]$ so that 
for each $0\leq j\leq m$, $\lambda_j' = \sum_{c_i\in\mathsf{pre}_j} \lambda_i \vec{p}_i(c_j)$.
Since $\Gamma(\pi)$ is a cluster graph, the set $\mathsf{pre}_j$ (the immediate predecessors of $c_j$) is equal to the set of every corner that has the same color as $c_j$, and similarly for $\mathsf{post}_i$. This means that for every color $k$,
\[\sum_{c_j\in C_k} \lambda_j'=\sum_{c_j\in C_k} \sum_{c_i\in C_k} \lambda_i \vec{p}_i(c_j) = \sum_{c_i\in C_k} \lambda_i \left(\sum_{c_j\in C_k} \vec{p}_i(c_j)\right)=\sum_{c_i\in C_k} \lambda_i\,.\]
Then, as $\nu\in\slice_\pi(\vecw)$ and $\nu'\in\slice_\pi(\vecw')$ imply that for every color $k$, $\sum_{c_j\in C_k} \lambda_j' = w_k$ and $\sum_{c_i\in C_k} \lambda_i=w_k'$, respectively, we conclude that $\vecw=\vecw'$.
\end{proof}

On the other hand, the reachability relation is full on any given slice.
\begin{lemma}\label{lm:partComplete}
    Let $\pi$ be a robustly iterable region cycle around a region $r$.
    Let $\vecw$ be a weight vector.
    Then, for any pair $\nu,\nu'\in\slice_\pi(\vecw)$, there exists a run from $\nu$ to $\nu'$ following $\pi$.
\end{lemma}
\begin{proof}
Let $\nu$ and $\nu'$ be two valuations in the same slice.
By Corollary~\ref{cor:fog-distribution}, it is possible to re-distribute the corner weights of $\nu$ on the edges of the folded orbit graph in order to match the corner weights of $\nu'$: indeed, they both have the same total weight on the corners of each individual color, and the folded orbit graph is complete on each color.
The way to build the probability distributions $\vec{p}_i$ for corners of a complete SCC is the same as in \cite{AB-formats11} %
for corners in a complete folded orbit graph.
\end{proof}

We note that whenever Lemma~\ref{lm:partComplete} applies, we have by Lemma~\ref{lm:slicedynamics} that for all $\nu,\nu'\in r$, there is a run from $\nu$ to $\nu'$ following $\pi$ if and only if there is a weight vector $\vecw$ so that $\nu,\nu'\in\slice_\pi(\vecw)$.

\begin{lemma}\label{lm:sliceShrink}
    Let $\pi$ be a region cycle around a region $r$.
    Let $\vecw$ be a weight vector.
    For any shrinking matrix $Q$ so that there exists $\delta>0$ such that
    $\slice_\pi(\vecw)-\delta Q\neq \emptyset$, we have for every $\delta>0$ that $\slice_\pi(\vecw)-\delta Q = (r-\delta Q)\cap\slice_\pi(\vecw)$.
\end{lemma}
\begin{proof}
    First, since $\slice_\pi(\vecw)\subseteq r$ we have $\slice_\pi(\vecw)-\delta Q = (r\cap\slice_\pi(\vecw))-\delta Q$.
    Then, by Proposition~\ref{prop:slicezones}, every slice can be described as the intersection of $r$ and a set of linear equalities of the shape $x-y=c$ or $x=c$. These constraints cannot be shrunk without making $\slice_\pi(\vecw)-\delta Q$ empty. Therefore, any non-zero entry in the shrinking matrix $Q$ must happen on a constraint of $r$, so that
    $(r\cap\slice_\pi(\vecw))-\delta Q = (r-\delta Q)\cap\slice_\pi(\vecw)$, wich lets us conclude.
\end{proof}

\newpage

\contWins*

\begin{proof}%
Given a robust region lasso $\pi_0\pi^\omega$,
so that $\pi_0$ starts from $(\ell_0,\{\mathbf{0}\})$, $\pi$ is a cycle around a region state $(\ell,r)$
with $\ell$ a winning state for the Büchi objective, $\pi_0$ is robust and $\pi$ is robustly iterable,
we argue that:
\begin{enumerate}
    \item Fix $N,\nu_N,\varepsilon$ by Lemma~\ref{lm:initslice} so that $\Ball_{d_\infty}(\nu_N,\varepsilon)\cap r \subseteq N$
    and for any shrinking matrix $Q$ so that there exists $\delta>0$ with $r-\delta Q\neq \emptyset$, we have $\exists \delta'>0$, $N\subseteq r-\delta' Q$.
    \item Let us show that $N$ is winning for $\cont$. It is enough to show that for all $\nu\in N$, there is a set $\mathcal N_\nu\subseteq r$ so that $\nu\in\mathcal N_\nu$ and $\mathcal N_\nu \subseteq\CPre^\delta_\pi(\mathcal N_\nu)$, \textit{i.e.}~$\cont$ has a strategy that remains in $\mathcal N_\nu$ when starting from an valuation in $\mathcal N_\nu$, so that it can iterate $\pi$ forever from $\nu$.
   \item Fix any valuation $\nu\in N$. Let $\vecw$ be the weight vector so that $\nu\in\slice_\pi(\vecw)$ by Lemma~\ref{lm:slicepartition}.
   Let $\mathcal N_\nu=N\cap\slice_\pi(\vecw)$. It is a DBM by Proposition~\ref{prop:slicezones}.
   \item We show that for small enough $\delta>0$, $\mathcal N_\nu \subseteq\CPre^\delta_\pi(\mathcal N_\nu)$.
   \begin{enumerate}
        \item Use Lemma~\ref{lm:partComplete}, $\nu\in\slice_\pi(\vecw)$ and $\nu\in\mathcal N_\nu$ to show that: \[\slice_\pi(\vecw)=\CPre_\pi^0(\mathcal N_\nu)\]
        \item Use Lemma~\ref{lm:cpredeltaslice} to show that there exists a shrinking matrix $Q$ 
        so that:
        \[\forall\delta,\CPre^\delta_\pi(\mathcal N_\nu)=\slice_\pi(\vecw)-\delta Q\]
        \item Use Lemma~\ref{lm:ballslice} to show $\slice_\pi(\vecw)-\delta Q\neq \emptyset$
        \item By Lemma~\ref{lm:sliceShrink}, we have $\slice_\pi(\vecw)-\delta Q = (r-\delta Q)\cap\slice_\pi(\vecw)$.
        \item By definition of $N$, for small enough $\delta>0$, $N\subseteq r-\delta Q$, so that  \[N\cap\slice_\pi(\vecw)\subseteq (r-\delta Q)\cap\slice_\pi(\vecw)\]
        \item Conclude that there exists $\delta>0$ such that
        \[\mathcal N_\nu\subseteq \CPre^\delta_{\pi}(\mathcal N_\nu)\]
    \end{enumerate}
\end{enumerate}
This means that starting from $\nu\in\mathcal N_\nu$, Controller has a strategy that always remain inside $\mathcal N_\nu$, at each iteration of $\pi$.
Therefore, controller can robustly follow $\pi^\omega$ when starting from any valuation $\nu$ in $N$.
Moreover, by Lemma~\ref{lm:ballslice}, $\CPre^\delta_{\pi_0}(N)$ contains the only valuation in the initial region $\mathbf{0}$. Hence, Controller has a winning strategy that robustly follows the lasso $\pi_0\pi^\omega$.
\end{proof}

\newpage
\section{Proof of Proposition~\ref{prop:pertWins}}
\label{app:pertWins}
Let $\nu$ be a valuation in $r$ of dimension $m$ such that $\nu$ has corner weights $\vecl=(\lambda_0,\dots,\lambda_m)$ in $r$.
Let $I$ be a subset of corner of $r$.
We define the function 
\begin{align*}
    L_I :~ &\overline{r} \to \Realnn,\\
    &\nu \mapsto \sum_{c_i \in I} \lambda_i c_i.
\end{align*}
We will study the dynamics of $L_I$ for region cycles whose folded orbit graph is not a cluster graph. 

For a fixed $\delta > 0$, we define the strategy $\sigma_\delta^p$ for $\pert$ as follows: 
After each delay $\nu \xrightarrow{d} \nu'$ played by \cont, consider a region $r$ such that $\nu' + [\alpha, \beta]\subseteq r$ for some $0<\alpha<\beta<\delta$ where $\beta - \alpha \geq \frac{\delta}{|\Clocks| +1}$. Such a region must exist by definition of the game, otherwise $\cont$ played an illegal move. $\pert$ then applies a perturbation of $\frac{(\beta-\alpha)}{2}$. 
This strategy guarantees a time progress of at least $\epsilon = \frac{\delta}{2(|\Clocks| + 1)}$.

A close inspection of the proof of Lemma~8.5.15 of \cite{Sankur13} shows that both these observations hold in our case. We restate this lemma in our setting in the following lemma:

\begin{lemma}\cite{SBMR-concur13}
\label{lm:LiapBound}
Let $\rho$ be a prefix of a run in $\outcome(-, \sigma_\delta^P)$ such that:
\begin{itemize}
    \item $\pi$ the projection of $\rho$ over regions is a cycle with at least one non-punctual edge,
    \item $\Gamma(\pi)$ contains at least two weakly connected SCCs $I$ and $J$,
    \item $\rho$ starts from $\nu$ and ends in $\nu'$,
\end{itemize}  
then the following equation holds true:
\[
  L_I(\nu') \leq L_I(\nu) - \frac{\epsilon^2}{2}
  \enspace,
\]
\end{lemma} 
The above lemma entails the intuition that if $\cont$ tries to iterate over $\pi$ to enforce $\rho$, then they will eventually be forced to propose delays smaller that $\delta$ which is not allowed by definition of our game semantics. This is formalized in the following corollary.

\begin{corollary}\label{cor:not-robust}
  Let $\rho$ be a run such that its projection over regions is a lasso $\pi_0\pi^\omega$ and assume that for any $k>0$, $\Gamma(\pi^k)$ is not robustly iterable, then $\rho$ cannot be in $\outcome(-,\sigma_\delta^P)$.
\end{corollary}
\begin{proof}
  Assume towards a contradiction that $\rho$ is in $\outcome(-,\sigma_\delta^P)$.
  By Lemma~\ref{lm:robCycle}, there exists $k > 0$ such that 
  $\Gamma(\pi^k)$ contains at least two weakly connected SCCs.
  Let $I$ be an initial SCCs in $\Gamma(\pi^k)$, and let $(\ell, \nu)$ be a state visited by $\rho$ such that $\nu$ is in the region spanned by the corners of $\Gamma(\pi)$.
  Now denote $\nu_n$ the valuation along $\rho$ after $n$ iteration of the cycle described by $\Gamma(\pi^k)$
  by Lemma~\ref{lm:LiapBound} we have:
  \begin{align*}
      0 \le L_I(\nu_n) < L_I(\nu) - n\frac{\epsilon^2}{2}
      \enspace.
  \end{align*}
  By definition of $\sigma_\delta^P$, 
  $n$ can be chosen such that: 
  \[
  n\frac{\epsilon^2}{2} > 1
  \enspace.
  \]
  This contradicts the non negativity of $L_I(\nu_n)$.
\end{proof}

In order to conclude we still need to argue that $\cont$ cannot win by 
switching between different non robustly iterable cycles. 
To see that, take any infinite run $\rho$. 
Since there exists a finite number of folded orbit graphs, using Ramsey like arguments
(c.f. Theorem~4.5.2 of~\cite{Sankur13}) we get that the projection over regions of the run $\rho$ can be decomposed as $\pi_0\pi_1\pi_2 \ldots \pi_n \ldots$ such that
$\Gamma(\pi_1) = \Gamma(\pi_2) = \Gamma(\pi_n) = \ldots$.
By assumption, $\Gamma(\pi_1)$ cannot be robustly iterable, hence Corollary \ref{cor:not-robust} implies that $\rho$ cannot be in $\outcome(-,\sigma_\delta^P)$.
In other words, for any $\cont$ strategy $\sigmacont$, $\outcome(\sigmacont, \sigma_\delta^P)$ cannot be at the same time an infinite run and a run that visits locations in $\mathsf{Buchi}$ infinitely often,
hence $\sigmacont$ is not winning for the Büchi objective.

\end{document}